\documentclass[journal]{IEEEtranTIE}
\usepackage{graphicx}
\pdfoutput=1
\usepackage{cite}
\usepackage{picinpar}
\usepackage{amsmath}
\usepackage{amssymb}
\usepackage{url}
\usepackage{flushend}
\usepackage[latin1]{inputenc}
\usepackage{colortbl}
\usepackage{soul}
\usepackage{multirow}
\usepackage{pifont}
\usepackage{color}
\usepackage{alltt}
\usepackage[hidelinks]{hyperref}
\usepackage{enumerate}
\usepackage{siunitx}
\usepackage{breakurl}
\usepackage{epstopdf}
\usepackage{pbox}
\usepackage{subfigure}
\usepackage{balance}

\newtheorem{assumption}{Assumption}
\newtheorem{lemma}{Lemma}
\newtheorem{remark}{Remark}

\newtheorem{theorem}{Theorem}

\newenvironment{proof}{{\indent \indent \it Proof:}}

\begin{document}
\title{	Globally Intelligent Adaptive Finite-/Fixed- Time Tracking Control for Strict-Feedback Nonlinear Systems via Composite Learning Approaches}

\author{
	\vskip 1em
	
	Xidong Wang

	\thanks{
					
		Xidong Wang is with the Research Institute of Intelligent Control and Systems, School of Astronautics, Harbin Institute of Technology, Harbin 150001, China (e-mail: 17b904039@stu.hit.edu.cn). 
	}
}

\maketitle
	
\begin{abstract}
This article focuses on the globally composite adaptive law-based intelligent finite-/fixed- time (FnT/FxT) tracking control issue for a family of uncertain strict-feedback nonlinear systems. First, intelligent approximators with  new composite updating laws are developed to model uncertain nonlinear terms, which encompass prediction errors to enhance intelligent approximators' learning behaviors and fewer online learning parameters to diminish computational burden. Then, a novel smooth switching function coupled with robust controllers is designed to pull system states back when the transients are out of the approximators' active domain. After that, a modified FnT/FxT backstepping technique is constructed to render output to follow the reference trajectory, and an adaptive law is employed to alleviate the impact of external disturbances. Moreover, input nonlinearities are considered in the design of control laws. In view of the known actuator magnitude and rate saturations, a novel adaptive fast FxT auxiliary variable system is established to effectively diminish saturation time, and further alleviate the impact of the auxiliary signal on the tracking error. Furthermore, when the input is subject to unknown magnitude saturation, rate saturation and dead-zone simultaneously, a smooth function is introduced to approximate the non-smooth input nonlinearities, where the unknown parameters are modeled via the presented global intelligent approximators. It is theoretically confirmed that the proposed control strategies ensure globally FnT/FxT boundedness of all the closed-loop variables. Finally, the validity of theoretical results is testified via a simulation case.
\end{abstract}


{}

\definecolor{limegreen}{rgb}{0.2, 0.8, 0.2}
\definecolor{forestgreen}{rgb}{0.13, 0.55, 0.13}
\definecolor{greenhtml}{rgb}{0.0, 0.5, 0.0}

\section{Introduction}

\IEEEPARstart{I}{n} the past few decades, the tracking control issue of nonlinear systems with strict feedback has been an active topic. As one of the most powerful control design techniques, backstepping control framework provides a systematic solution for tracking control of strict-feedback systems \cite{DSC,CFB,ACFB}. To tackle complex uncertainties in the system, intelligent approximators, embracing neutral network (NN) and fuzzy logic system (FLS), are introduced into the backstepping technique \cite{neural1,neural2,fuzzy1,fuzzy2}. Due to the high correlation between tracking performance and intelligent estimation performance, growing research attention has been concentrated on the precision of neural/fuzzy approximators in modeling unknown dynamics recently. The authors of \cite{composite1} developed a composite NN-based backstepping algorithm, where the prediction error extracted from serial-parallel estimation model (SPEM) was added to the neural adaptive law for obtaining better approximation performance. Subsequently, this approach of constructing prediction errors to reinforce estimation capabilities was applied to various control schemes to enhance control performance \cite{composite2,composite3}. It is worth noting that the prerequisite for excellent estimation performance is that the intelligent approximators should retain valid at all times. Once a transient transcends the active region of the intelligent approximator, tracking performance slumps, and even the system becomes unstable. To fulfill the global stability of intelligent control, a switching protocol with extra robust control laws was designed to drive transients back to the neural working domain \cite{global1}. Afterwards, a multitude of improved globally neural/fuzzy control strategies were established to diminish input energy and attain superior intelligent approximation ability \cite{global2,global3,global4,global5,global6}. However, most of the above-mentioned approaches are infinite-time stability. 

To acquire rapid transient response and strong robustness in practical tracking control, the FnT/FxT backstepping techniques are proposed \cite{fnt1,fnt2,fnt3,fnt4,fnt5,fnt_own,fxt2,fxt3,fxt4,fxt5}. Additionally, the settling time boundary of the FxT backstepping approach is irrespective of initial conditions \cite{fxt1,fxt6}. Moreover, the authors of \cite{F_FxT} proposed a practical fast FxT (F-FxT) control strategy. In \cite{fnt_composite1}, a composite neural FnT controller was developed, which introduced prediction errors to strengthen RBFNNs' capability of identifying uncertainties. The article \cite{fnt_composite2} presented a switching control strategy between a prediction error-based intelligent control law and a robust controller to achieve globally FnT stability, where the prediction error is extracted from a modified SPEM. The authors in \cite{fxt_composite1} designed a neural adaptive FxT backstepping algorithm with composite learning laws to enhance NNs' performance. Nevertheless, the computational burden of the composite adaptive law updating is heavy, since the number of composite update weights is closely related to the number of neurons in the RBFNN.

When designing controllers for physical systems, it is a common phenomenon to encounter input nonlinearities, embracing saturation, backlash and dead-zone. To deal with the issue of input saturations with known parameters, various kinds of auxiliary signal compensators were exploited \cite{sat_a1,sat_a2,sat_f1,sat_f2}. The authors of \cite{sat_fa} analyzed the impact of compensation system gain on the saturation compensation performance in detail, and then developed an innovative adaptive FxT anti-saturation compensator. For unknown input nonlinearities, smooth functions were constructed to model the non-smooth actuator nonlinearities, and the unknown parameters were estimated by adopting intelligent approximators \cite{sat_NN1,sat_NN2,sat_zone_fuzzy1,sat_zone_fuzzy2}. 

Motivated by the preceding statements, this article establishes globally composite adaptive intelligent FnT/FxT backstepping control schemes for the tracking control of uncertain strict-feedback nonlinear systems subject to input nonlinearities and external disturbances. First, intelligent approximators with novel composite learning laws are designed to identify uncertain nonlinear functions. Then, a new smooth switching mechanism is constructed to drive system states back when the transients are out of the approximators' valid domain. Next, an improved FnT/FxT backstepping protocol is developed to enable the output to track the reference variable. Stability analysis implicates that all signals in the closed-loop system are globally FnT/FxT bounded. The core contributions of this article are listed as follows:
\begin{enumerate}[1)]
	\item New composite learning approaches based on prediction errors are developed for weight updating of intelligent approximators with fewer online learning parameters, which can simultaneously enhance tracking performance and alleviate computational burden.
	\item A novel smooth switching function coupled with robust controllers is developed to bring system states back when the transients are out of the approximators' working domain. In addition, the condition that the boundaries of the uncertain nonlinear terms require to be known in advance is relaxed.
    \item The multiple input nonlinearities are taken into account in the design of controllers. When the input contains known magnitude and rate saturations, a novel adaptive F-FxT auxiliary variable system is designed to effectively diminish saturation time, and further weaken the impact of the auxiliary signal on the tracking error. When the input is simultaneously subject to unknown magnitude saturation, rate saturation and dead-zone, a smooth function is introduced to approximate the non-smooth actuator nonlinearities, where the unknown parameters are modeled by means of the presented global intelligent approximators.
\end{enumerate}

In this article, we select radial basis function neural network (RBFNN) as the intelligent approximator, while the control strategies obtained by replacing RBFNN with FLS are still valid.

Notation: $\left\|  \cdot  \right\|$ represents the Euclidean norm.
%2-A
\section{Problem formulation and lemmas}
\subsection{Problem Formulation}
Consider a class of strict-feedback nonlinear systems with the following dynamics
\begin{equation}
\left\{  \begin{aligned}
&{{\dot \rho }_i} = {h_i}\left( {{{\bar \rho }_i}} \right) + {g_i}\left( {{{\bar \rho }_i}} \right){\rho _{i + 1}} + {d_i},i = 1,2, \ldots ,n - 1\\
&{{\dot \rho }_n} = {h_n}\left( {{{\bar \rho }_n}} \right) + {g_n}\left( {{{\bar \rho }_n}} \right)u + {d_n}\\
&y = {\rho _1}
 \end{aligned} \right.
\end{equation}
where ${\bar \rho _i} = {\left[ {{\rho _1},{\rho _2}, \ldots ,{\rho _i}} \right]^T} \in {\mathbb{R}}^i$,$\left( {i = 1,2, \ldots ,n} \right)$, stand for the state vectors. $y \in {\mathbb{R}}$ and $u \in {\mathbb{R}}$ are the system output and input variables, respectively. ${g_i}\left(  \cdot  \right)$ denote the known nonzero control gains with $0 < \underline g < \left| {{g_i}\left(  \cdot  \right)} \right| < \bar g$ (Without loss of generality, assume ${g_i}\left(  \cdot  \right)>0$). The unknown uncertainties ${h_i}\left(  \cdot  \right)$ are smooth. $d_i$ represent unknown bounded external disturbances. 

The control mission of this article is to establish globally composite adaptive law-based neural FnT/FxT control strategies to ensure that all signals in the closed-loop system are globally FnT/FxT bounded, and the tracking error tends to an arbitrarily small domain near zero for finite/fixed time.  
\begin{assumption}
The given reference variable ${y_r}(t)$ and its first derivative ${\dot y_r}(t)$ are bounded and smooth.
\end{assumption}
\begin{assumption}
For uncertain nonlinearities ${h_i}\left( {{{\bar \rho }_i}} \right)$, there exist known positive smooth functions ${H_i}\left( {{{\bar \rho }_i}} \right)$ and unknown constants ${\tau _i} \ge 0$ satisfying $\left| {{h_i}\left( {{{\bar \rho }_i}} \right)} \right| \le {\tau _i}{H_i}\left( {{{\bar \rho }_i}} \right)$.
\end{assumption}
\begin{remark}
If we select ${H_i}\left( {{{\bar \rho }_i}} \right) \equiv 1$, then the assumption of uncertain nonlinearities becomes $\left| {{h_i}\left( {{{\bar \rho }_i}} \right)} \right| \le {\tau _i}$, which is a common assumption condition for system (1).
\end{remark}
%2-B
\subsection{Lemmas}
\begin{lemma} [\cite{fxt2}]
The following inequalities hold for any ${x_s} \in {\mathbb{R}},s = 1,2, \ldots ,p$
\begin{equation}
\begin{aligned}
{\left( {\sum\limits_{s = 1}^p {\left| {{x_s}} \right|} } \right)^{q_1}} \le \sum\limits_{s = 1}^p {{{\left| {{x_s}} \right|}^{q_1}}},&\quad 0 < q_1 \le 1\\
{p^{1 - q_2}}} {\left( {\sum\limits_{s = 1}^p {\left| {{x_s}} \right|} } \right)^{q_2}} \le \sum\limits_{s = 1}^p {{{\left| {{x_s}} \right|}^{q_2}},&\quad q_2 > 1
\end{aligned}
\end{equation}
\end{lemma}

\begin{lemma} [\cite{lemma2}]
For $W \in {\mathbb{R}}$ and $\eta >0$, it holds that
\begin{equation}
0 \le \left| W \right| - W\tanh \left( {\frac{W}{\eta }} \right) \le 0.2785\eta  
\end{equation}
\end{lemma}

\begin{lemma} [\cite{lemma3}]
For any ${\Upsilon _1},{\Upsilon _2} \in {\mathbb{R}}$, and positive constants $s_1,s_2$, one gets  
\begin{equation}
{\left| {{\Upsilon _1}} \right|^{{s_1}}}{\left| {{\Upsilon _2}} \right|^{{s_2}}} \le \frac{{{s_1}}}{{{s_1}{\rm{ + }}{s_2}}}{\left| {{\Upsilon _1}} \right|^{{s_1}{\rm{ + }}{s_2}}}{\rm{ + }}\frac{{{s_2}}}{{{s_1}{\rm{ + }}{s_2}}}{\left| {{\Upsilon _2}} \right|^{{s_1}{\rm{ + }}{s_2}}}
\end{equation}
\end{lemma}

\begin{lemma} [\cite{fnt_own}]
For any $\lambda \in {\mathbb{R}}$ and positive constants $\mu ,\kappa $, there holds
\begin{equation}
0 \le \left| \lambda  \right| - {\lambda ^2}\sqrt {\frac{{{\lambda ^2} + {\mu ^2} + {\kappa ^2}}}{{\left( {{\lambda ^2} + {\mu ^2}} \right)\left( {{\lambda ^2} + {\kappa ^2}} \right)}}}  < \frac{{\mu \kappa }}{{\sqrt {{\mu ^2} + {\kappa ^2}} }} 
\end{equation}
\end{lemma}

To achieve smooth switching, we design the following novel $nth$-order continuous differentiable functions: 
\begin{equation}
{\varpi _j}\left( {{\rho _j}} \right) = 
\left\{
\begin{aligned}
&1,&\left| {{\rho _j}} \right| \le {c_{j1}}\\
&{\cos ^{n + 1}}\left( {\frac{\pi }{2}{{\left( {\frac{{\rho _j^2 - c_{j1}^2}}{{c_{j2}^2 - c_{j1}^2}}} \right)}^n}} \right),&otherwise\\
&0,&\left| {{\rho _j}} \right| \ge {c_{j2}}
\end{aligned} 
\right.
\end{equation}
and
\begin{equation}
{\bar \varpi _j}\left( {{\rho _j}} \right) = 
\left\{
\begin{aligned}
&1,&\left| {{\rho _j}} \right| \le {c_{j1}}\\
&{\cos ^{n + 1}}\left( {\frac{\pi }{2}{{\left( {\frac{{\left| {{\rho _j}} \right| - {c_{j1}}}}{{{c_{j2}} - {c_{j1}}}}} \right)}^n}} \right),&otherwise\\
&0,&\left| {{\rho _j}} \right| \ge {c_{j2}}
\end{aligned} 
\right.
\end{equation}
where $c_{j1}>0$ and $c_{j2}>0$ stand for the left and right boundaries of RBFNNs' valid domain, $j = 1,2, \ldots ,i$.

Then, the switch indicator functions for identifying whether the state variables are located in the active region of RBFNNs is introduced as follows: ${w_i}\left( {{{\bar \rho }_i}} \right) = \prod\limits_{j = 1}^i {{\varpi _j}\left( {{\rho _j}} \right)}$ or ${w_i}\left( {{{\bar \rho }_i}} \right) = \prod\limits_{j = 1}^i {{{\bar \varpi }_j}\left( {{\rho _j}} \right)}$.
\begin{remark}
The switch indicator function is developed to enable smooth switching between the composite NN controller and the robust control law, which guarantees that only the composite adaptive law is valid when the state variables are within the RBFNN's working region to save control energy.
\end{remark}

Enlightened by \cite{2021Improved,F_FxT}, we propose the following \emph{Lemma}:
\begin{lemma} 
If there exist a continuous Lyapunov function $V\left( \theta  \right)$ and constants ${\theta _1} > 0,{\theta _2} > 0,{\theta _3} > 0,0 < {\theta _4} < \infty ,0 < m < 1,r > 1$ such that: $\dot V(\rho ) \le  - {\theta _1}{V^m}(\rho ) - {\theta _2}{V^r}(\rho ) - {\theta _3}V(\rho ) + {\theta _4}$, then the trajectory of system $\dot \rho  = f(\rho )$ is practical F-FxT stable, and the residual set is formulated as:
\begin{equation}\small
\left\{ {\rho :V\left( \rho  \right) \le \min \left\{ {{{\left( {\frac{{{\theta _4}}}{{\left( {1 - \nu } \right){\theta _1}}}} \right)}^{\frac{1}{m}}},{{\left( {\frac{{{\theta _4}}}{{\left( {1 - \nu } \right){\theta _2}}}} \right)}^{\frac{1}{r}}},\frac{{{\theta _4}}}{{\left( {1 - \nu } \right){\theta _3}}}} \right\}} \right\}
\end{equation}
where $\nu  \in (0,1)$. Additionally, the reaching time $T$ is bounded by $T \le \max \left\{ {{T_1},{T_2},{T_3}} \right\}$, where 
\begin{equation}
\begin{aligned}
&{T_1} = \frac{\phi }{{{\theta _3}(r - 1)}}\ln \left( {1 + \frac{{{\theta _3}}}{{\nu {\theta _1}}}{{\left( {\frac{{\nu {\theta _1}}}{{{\theta _2}}}} \right)}^{\frac{1}{\phi }}}} \right)\\
&{T_2} = \frac{\phi }{{{\theta _3}(r - 1)}}\ln \left( {1 + \frac{{{\theta _3}}}{{{\theta _1}}}{{\left( {\frac{{{\theta _1}}}{{\nu {\theta _2}}}} \right)}^{\frac{1}{\phi }}}} \right)\\
&{T_3} = \frac{\phi }{{\nu {\theta _3}(r - 1)}}\ln \left( {1 + \frac{{\nu {\theta _3}}}{{{\theta _1}}}{{\left( {\frac{{{\theta _1}}}{{{\theta _2}}}} \right)}^{\frac{1}{\phi }}}} \right)
\end{aligned}
\end{equation}
with $\phi  = {{\left( {r - m} \right)} \mathord{\left/
 {\vphantom {{\left( {r - m} \right)} {\left( {1 - m} \right)}}} \right.
 \kern-\nulldelimiterspace} {\left( {1 - m} \right)}}$.
\end{lemma}

If $m+r=2$, then a tighter approximation of the reaching time for practical F-FxT stability is given as below
\begin{lemma} 
If $m+r=2$, the reaching time $T$ of \emph{Lemma 5} is further formulated as follows:

(i) If the residual set is $\left\{ {\rho :V\left( \rho  \right) \le {{\left( {\frac{{{\theta _4}}}{{\left( {1 - \nu } \right){\theta _1}}}} \right)}^{\frac{1}{m}}}} \right\}$, then $T \le \frac{1}{{r - 1}}{T_{\max 1}}$, where
\begin{equation}\small
{T_{\max 1}} = \left\{\begin{aligned}
&\frac{2}{{\sqrt {{\Delta _1}} }}\left( {\frac{\pi }{2} - \arctan \frac{{2{\theta _2}{\chi _1} + {\theta _3}}}{{\sqrt {{\Delta _1}} }}} \right),&{\theta _3} < 2\sqrt {\nu {\theta _1}{\theta _2}} \\
&\frac{2}{{2{\theta _2}{\chi _1} + {\theta _3}}},&{\theta _3} = 2\sqrt {\nu {\theta _1}{\theta _2}} \\
&\frac{1}{{\sqrt { - {\Delta _1}} }}\ln \left( {\frac{{2{\theta _2}{\chi _1} + {\theta _3} + \sqrt { - {\Delta _1}} }}{{2{\theta _2}{\chi _1} + {\theta _3} - \sqrt { - {\Delta _1}} }}} \right),&{\theta _3} > 2\sqrt {\nu {\theta _1}{\theta _2}} 
\end{aligned} \right.
\end{equation}
with ${\chi _1} = {\left( {\frac{{{\theta _4}}}{{\left( {1 - \nu } \right){\theta _1}}}} \right)^{\frac{{r - 1}}{m}}}, {\Delta _1} = 4\nu {\theta _1}{\theta _2} - \theta _3^2$.

(ii) If the residual set is $\left\{ {\rho :V\left( \rho  \right) \le {{\left( {\frac{{{\theta _4}}}{{\left( {1 - \nu } \right){\theta _2}}}} \right)}^{\frac{1}{r}}}} \right\}$, then $T \le \frac{1}{{r - 1}}{T_{\max 2}}$, where
\begin{equation}\small
{T_{\max 2}} = \left\{ \begin{aligned}
&\frac{2}{{\sqrt {{\Delta _2}} }}\left( {\frac{\pi }{2} - \arctan \frac{{2\nu {\theta _2}{\chi _2} + {\theta _3}}}{{\sqrt {{\Delta _2}} }}} \right),&{\theta _3} < 2\sqrt {\nu {\theta _1}{\theta _2}} \\
&\frac{2}{{2\nu {\theta _2}{\chi _2} + {\theta _3}}},&{\theta _3} = 2\sqrt {\nu {\theta _1}{\theta _2}} \\
&\frac{1}{{\sqrt { - {\Delta _2}} }}\ln \left( {\frac{{2\nu {\theta _2}{\chi _2} + {\theta _3} + \sqrt { - {\Delta _2}} }}{{2\nu {\theta _2}{\chi _2} + {\theta _3} - \sqrt { - {\Delta _2}} }}} \right),&{\theta _3} > 2\sqrt {\nu {\theta _1}{\theta _2}} 
\end{aligned} \right.
\end{equation}
with ${\chi _2} = {\left( {\frac{{{\theta _4}}}{{\left( {1 - \nu } \right){\theta _2}}}} \right)^{\frac{{r - 1}}{r}}},{\Delta _2} = 4\nu {\theta _1}{\theta _2} - \theta _3^2$.

(iii) If the residual set is $\left\{ {\rho :V\left( \rho  \right) \le \frac{{{\theta _4}}}{{\left( {1 - \nu } \right){\theta _3}}}} \right\}$, then $T \le \frac{1}{{r - 1}}{T_{\max 3}}$, where
\begin{equation}\small
{T_{\max 3}} = \left\{ \begin{aligned}
&\frac{2}{{\sqrt {{\Delta _3}} }}\left( {\frac{\pi }{2} - \arctan \frac{{2{\theta _2}{\chi _3} + \nu {\theta _3}}}{{\sqrt {{\Delta _3}} }}} \right),&{\theta _3} < \frac{2}{\nu }\sqrt {{\theta _1}{\theta _2}} \\
&\frac{2}{{2{\theta _2}{\chi _3} + \nu {\theta _3}}},&{\theta _3} = \frac{2}{\nu }\sqrt {{\theta _1}{\theta _2}} \\
&\frac{1}{{\sqrt { - {\Delta _3}} }}\ln \left( {\frac{{2{\theta _2}{\chi _3} + \nu {\theta _3} + \sqrt { - {\Delta _3}} }}{{2{\theta _2}{\chi _3} + \nu {\theta _3} - \sqrt { - {\Delta _3}} }}} \right),&{\theta _3} > \frac{2}{\nu }\sqrt {{\theta _1}{\theta _2}} 
\end{aligned} \right.
\end{equation}
with ${\chi _3} = {\left( {\frac{{{\theta _4}}}{{\left( {1 - \nu } \right){\theta _3}}}} \right)^{r - 1}},{\Delta _3} = 4{\theta _1}{\theta _2} - {\left( {\nu {\theta _3}} \right)^2}$.
\end{lemma}
%3-A
\section{Globally composite neural FnT/FxT controllers without input nonlinearities}
\subsection{Globally neural FnT controller with composite learning laws}
Define tracking errors as below
\begin{equation}
\begin{aligned}
&{\zeta _1} = {\rho _1} - {y_r}\\
&{\zeta _i} = {\rho _i} - {\rho _{i,c}},\quad i = 2,3, \ldots ,n
\end{aligned}
\end{equation}
where ${\rho _{i,c}}$ are obtained by the rapid FnT differentiator \cite{lemma6} with virtual control variables ${\alpha _{i - 1}}$ being inputs.

To alleviate the influence of $\left( {{\rho _{i,c}} - {\alpha _{i - 1}}} \right)$, the following compensation system is adopted [19]
\begin{equation}
\begin{aligned}
&{{\dot \sigma }_1} =  - {k_1}{\sigma _1} + {g_1}\left( {{\rho _{2,c}} - {\alpha _1}} \right) + {g_1}{\sigma _2} - {\gamma _1}\sigma _1^m\\
&{{\dot \sigma }_i} =  - {k_i}{\sigma _i} + {g_i}\left( {{\rho _{\left( {i + 1} \right),c}} - {\alpha _i}} \right) - {g_{i - 1}}{\sigma _{i - 1}} + {g_i}{\sigma _{i + 1}} - {\gamma _i}\sigma _i^m\\
&{{\dot \sigma }_n} =  - {k_n}{\sigma _n} - {g_{n - 1}}{\sigma _{n - 1}} - {\gamma _n}\sigma _n^m
\end{aligned}
\end{equation}
where ${\sigma _i}\left( 0 \right) = 0$, and ${k_i} > 0,{\gamma _i} > 0,\frac{1}{2} < m = \frac{{{m_2}}}{{{m_1}}} < 1$ are designed parameters with $m_1, m_2$ being odd numbers.

In terms of the universal approximation theory \cite{RBF2014}, if ${\bar \rho _i}$ are on the neural active compact set, then ${h_i}\left( {{{\bar \rho }_i}} \right) = l_i^T{\psi _i}\left( {{{\bar \rho }_i}} \right) + {\varepsilon _i}$, where $l_i$ is the ideal weight vector and ${\psi _i}\left( {{{\bar \rho }_i}} \right)$ is the Gaussian basis function vector, $\left| {{\varepsilon _i}} \right| \le {\bar \varepsilon _i},{L_i} = {\left\| {{l_i}} \right\|^2}$. 

Then, virtual control variables ${\alpha _i}$ and the control input $u = {\alpha _n}$ are developed as follows (Method 1)
\begin{equation}\small
\begin{aligned}
{\alpha _1} &= \frac{1}{{{g_1}}}\left[ { - {k_1}{\zeta _1} + {{\dot y}_r} - {p_1}{\varphi _1}\left( {{\lambda _1}} \right) - {{\hat d}_1}\tanh \left( {\frac{{{\lambda _1}}}{{{\eta _{1d}}}}} \right)} \right]\\
 &- \frac{1}{{{g_1}}}\left[ {\frac{{{w_1}\left( {{{\bar \rho }_1}} \right){\lambda _1}}}{{2a_1^2}}{{\hat L}_1}\psi _1^T{\psi _1} + \left( {1 - {w_1}} \right){{\hat \tau }_1}{H_1}\tanh \left( {\frac{{{H_1}{\lambda _1}}}{{{\eta _1}}}} \right)} \right]\\
{\alpha _i} &= \frac{1}{{{g_i}}}\left[ { - {k_i}{\zeta _i} + {{\dot \rho }_{i,c}} - {g_{i - 1}}{\zeta _{i - 1}} - {p_i}{\varphi _i}\left( {{\lambda _i}} \right) - {{\hat d}_i}\tanh \left( {\frac{{{\lambda _i}}}{{{\eta _{id}}}}} \right)} \right]\\
 &- \frac{1}{{{g_i}}}\left[ {\frac{{{w_i}\left( {{{\bar \rho }_i}} \right){\lambda _i}}}{{2a_i^2}}{{\hat L}_i}\psi _i^T{\psi _i} + \left( {1 - {w_i}} \right){{\hat \tau }_i}{H_i}\tanh \left( {\frac{{{H_i}{\lambda _i}}}{{{\eta _i}}}} \right)} \right]
\end{aligned}
\end{equation} 
where ${p_i} > 0,{a_i} > 0,{\eta _i} > 0,{\eta _{id}} > 0,{\lambda _i} = {\zeta _i} - {\sigma _i}$, 
\begin{equation}
{\varphi _i}\left( {{\lambda _i}} \right) = \lambda _i^{1 + 2m}\sqrt {\frac{{\lambda _i^{2 + 2m} + \mu _i^2 + \kappa _i^2}}{{\left( {\lambda _i^{2 + 2m} + \mu _i^2} \right)\left( {\lambda _i^{2 + 2m} + \kappa _i^2} \right)}}} 
\end{equation}
with ${\mu _i} > 0,{\kappa _i} > 0$ being design parameters, and ${\hat L_i},{\hat \tau _i},{\hat d_i}$ are the approximation of ${L_i},{\tau _i},{d_i}$, respectively, which are expressed as follows
\begin{equation}
\begin{aligned}
&{\dot{ \hat L}_i} = {\beta _{hi}}{w_i}\left( {\lambda _i^2 + {\beta _{zi}}z_{iN}^2} \right)\psi _i^T{\psi _i} - {\beta _{i1}}{{\hat L}_i} - {\beta _{i2}}\hat L_i^m\\
&{\dot {\hat \tau }_i} = {\delta _{i1}}\left( {1 - {w_i}} \right){H_i}{\lambda _i}\tanh \left( {\frac{{{H_i}{\lambda _i}}}{{{\eta _i}}}} \right) - {\delta _{i2}}{{\hat \tau }_i} - {\delta _{i3}}\hat \tau _i^m\\
&{\dot {\hat d}_i} = {q_{i1}}{\lambda _i}\tanh \left( {\frac{{{\lambda _i}}}{{{\eta _{id}}}}} \right) - {q_{i2}}{{\hat d}_i} - {q_{i3}}\hat d_i^m
\end{aligned}
\end{equation}
with ${\beta _{hi}},{\beta _{zi}},{\beta _{i1}},{\beta _{i2}},{\delta _{i1}},{\delta _{i2}},{\delta _{i3}},{q_{i1}},{q_{i2}},{q_{i3}}$ being positive design parameters.

The prediction errors ${z_{iN}} = {\rho _i} - {\hat \rho _i}$ are extracted from the following modified finite-time SPEM (denote ${\rho_{n+1}}=u$):
\begin{equation}
\begin{aligned}
{\dot{\hat \rho }_i} =& \frac{{{w_i}\left( {{{\bar \rho }_i}} \right)}}{{2a_i^2}}{{\hat L}_i}{z_{iN}}\psi _i^T{\psi _i} + \left( {1 - {w_i}} \right){{\hat \tau }_{iN}}{H_i}\tanh \left( {\frac{{{H_i}{z_{iN}}}}{{{\eta _{iN}}}}} \right)\\
 &+ {g_i}{\rho _{i+1}} + {r_{i1}}{z_{iN}} + {r_{i2}}z_{iN}^m + {{\hat d}_{iN}}\tanh \left( {\frac{{{z_{iN}}}}{{{\eta _{idN}}}}} \right)
\end{aligned}
\end{equation}
where
\begin{equation}\small
\begin{aligned}
&{\dot {\hat \tau }_{iN}} = {\delta _{i1N}}\left( {1 - {w_i}} \right){H_i}{z_{iN}}\tanh \left( {\frac{{{H_i}{z_{iN}}}}{{{\eta _{iN}}}}} \right) - {\delta _{i2N}}{{\hat \tau }_{iN}} - {\delta _{i3N}}\hat \tau _{iN}^m\\
&{\dot {\hat d}_{iN}} = {q_{i1N}}{{z_{iN}}}\tanh \left( {\frac{{{z_{iN}}}}{{{\eta _{idN}}}}} \right) - {q_{i2N}}{{\hat d}_{iN}} - {q_{i3N}}\hat d_{iN}^m
\end{aligned}
\end{equation}
with ${\eta _{iN}},{\eta _{idN}},{\delta _{i1N}},{\delta _{i2N}},{\delta _{i3N}},{q_{i1N}},{q_{i2N}},{q_{i3N}}$ being positive design parameters.

%Method 2-3
If we define ${N_i} = \max \left\{ {{{\bar \varepsilon }_i},\left\| {{l_i}} \right\|} \right\},{\psi _{ih}} = \left\| {{\psi _i}} \right\| + 1$ (Method 2) or ${N_i} = \left\| {{l_i}} \right\|,{\psi _{ih}} = \left\| {{\psi _i}} \right\|$ (Method 3), then another virtual control variables ${\alpha _i}$ and the control input $u = {\alpha _n}$ are constructed  as follows
\begin{equation}\small
\begin{aligned}
&{\alpha _1} = \frac{1}{{{g_1}}}\left[ { - {k_1}{\zeta _1} + {{\dot y}_r} - {p_1}{\varphi _1}\left( {{\lambda _1}} \right) - {{\hat d}_1}\tanh \left( {\frac{{{\lambda _1}}}{{{\eta _{1d}}}}} \right)} \right]\\
 &- \frac{1}{{{g_1}}}\left[ {{w_1}{{\hat N}_1}{\psi _{1h}}\tanh \left( {\frac{{{\psi _{1h}}{\lambda _1}}}{{{\eta _{1\theta }}}}} \right) + \left( {1 - {w_1}} \right){{\hat \tau }_1}{H_1}\tanh \left( {\frac{{{H_1}{\lambda _1}}}{{{\eta _1}}}} \right)} \right]\\
&{\alpha _i} = \frac{1}{{{g_i}}}\left[ { - {k_i}{\zeta _i} + {{\dot \rho }_{i,c}} - {g_{i - 1}}{\zeta _{i - 1}} - {p_i}{\varphi _i}\left( {{\lambda _i}} \right) - {{\hat d}_i}\tanh \left( {\frac{{{\lambda _i}}}{{{\eta _{id}}}}} \right)} \right]\\
 &- \frac{1}{{{g_i}}}\left[ {{w_i}{{\hat N}_i}{\psi _{ih}}\tanh \left( {\frac{{{\psi _{ih}}{\lambda _i}}}{{{\eta _{i\theta }}}}} \right) + \left( {1 - {w_i}} \right){{\hat \tau }_i}{H_i}\tanh \left( {\frac{{{H_i}{\lambda _i}}}{{{\eta _i}}}} \right)} \right]
\end{aligned}
\end{equation}
where ${\hat N_i},{\hat \tau _i},{\hat d_i}$ are the approximation of ${N_i},{\tau _i},{d_i}$ respectively, which are formulated as follows
\begin{equation}\small
\begin{aligned}
{\dot {\hat N}_i} = &{\beta _{hi}}{w_i}\left[ {{\lambda _i}{\psi _{ih}}\tanh \left( {\frac{{{\lambda _i}{\psi _{ih}}}}{{{\eta _{i\theta }}}}} \right) + {\beta _{zi}}{z_{iN}}{\psi _{ih}}\tanh \left( {\frac{{{z_{iN}}{\psi _{ih}}}}{{{\eta _{i\theta N}}}}} \right)} \right]\\
 &- {\beta _{i1}}{{\hat N}_i} - {\beta _{i2}}\hat N_i^m\\
{\dot {\hat \tau }_i} = &{\delta _{i1}}\left( {1 - {w_i}} \right){H_i}{\lambda _i}\tanh \left( {\frac{{{H_i}{\lambda _i}}}{{{\eta _i}}}} \right) - {\delta _{i2}}{{\hat \tau }_i} - {\delta _{i3}}\hat \tau _i^m\\
{\dot {\hat d}_i} =& {q_{i1}}{\lambda _i}\tanh \left( {\frac{{{\lambda _i}}}{{{\eta _{id}}}}} \right) - {q_{i2}}{{\hat d}_i} - {q_{i3}}\hat d_i^m
\end{aligned}
\end{equation}
with ${\eta _{1\theta }} > 0,{\eta _{1\theta N}} > 0$ being design parameters.

The prediction errors ${z_{iN}} = {\rho _i} - {\hat \rho _i}$ are extracted from the following modified finite-time SPEM (denote ${\rho_{n+1}}=u$):
\begin{equation}
\begin{aligned}
{\dot{\hat \rho }_i} = &{w_i}{{\hat N}_i}{\psi _{ih}}\tanh \left( {\frac{{{\psi _{ih}}{z_{iN}}}}{{{\eta _{i\theta N}}}}} \right) + \left( {1 - {w_i}} \right){{\hat \tau }_{iN}}{H_i}\tanh \left( {\frac{{{H_i}{z_{iN}}}}{{{\eta _{iN}}}}} \right)\\
& + {g_i}{\rho _{i+1}} + {r_{i1}}{z_{iN}} + {r_{i2}}z_{iN}^m + {{\hat d}_{iN}}\tanh \left( {\frac{{{z_{iN}}}}{{{\eta _{idN}}}}} \right)
\end{aligned}
\end{equation}
where
\begin{equation}\small
\begin{aligned}
&{\dot {\hat \tau }_{iN}} = {\delta _{i1N}}\left( {1 - {w_i}} \right){H_i}{z_{iN}}\tanh \left( {\frac{{{H_i}{z_{iN}}}}{{{\eta _{iN}}}}} \right) - {\delta _{i2N}}{{\hat \tau }_{iN}} - {\delta _{i3N}}\hat \tau _{iN}^m\\
&{\dot {\hat d}_{iN}} = {q_{i1N}}{{z_{iN}}}\tanh \left( {\frac{{{z_{iN}}}}{{{\eta _{idN}}}}} \right) - {q_{i2N}}{{\hat d}_{iN}} - {q_{i3N}}\hat d_{iN}^m
\end{aligned}
\end{equation}
\begin{theorem}
Considering the uncertain strict-feedback system (1), the compensation system (14), the prediction errors and modified finite-time SPEM (18) or (22), the virtual control law and controller (15) or (20), the adaptive laws (17) or (21), all signals in the closed-loop system are globally FnT bounded, and the tracking error tends to an arbitrarily small domain around zero for finite time. 
\end{theorem}
\begin{proof}
For controller (15), the Lyapunov function is designed as follows
\begin{equation}
\begin{aligned}
V = &\frac{1}{2}\sum\limits_{i = 1}^n {\left( {\lambda _i^2 + \frac{{\tilde L_i^2}}{{2a_i^2{\beta _{hi}}}} + {\beta _{zi}}z_{iN}^2 + \sigma _i^2} \right)} \\
& + \frac{1}{2}\sum\limits_{i = 1}^n {\left( {\delta _{i1}^{ - 1}\tilde \tau _i^2 + \delta _{i1N}^{ - 1}\tilde \tau _{iN}^2 + q_{i1}^{ - 1}\tilde d_i^2 + q_{i1N}^{ - 1}\tilde d_{iN}^2} \right)}, 
\end{aligned}
\end{equation}
and for controller (20), the Lyapunov function is developed as below
\begin{equation}
\begin{aligned}
V = &\frac{1}{2}\sum\limits_{i = 1}^n {\left( {\lambda _i^2 + \frac{{\tilde N_i^2}}{{{\beta _{hi}}}} + {\beta _{zi}}z_{iN}^2 + \sigma _i^2} \right)} \\
& + \frac{1}{2}\sum\limits_{i = 1}^n {\left( {\delta _{i1}^{ - 1}\tilde \tau _i^2 + \delta _{i1N}^{ - 1}\tilde \tau _{iN}^2 + q_{i1}^{ - 1}\tilde d_i^2 + q_{i1N}^{ - 1}\tilde d_{iN}^2} \right)}, 
\end{aligned}
\end{equation}
where ${\tilde L_i} = {L_i} - {\hat L_i},{\tilde N_i} = {N_i} - {\hat N_i},{\tilde \tau _i} = {\tau _i} - {\hat \tau _i},{\tilde d_i} = {\bar d_i} - {\hat d_i},{\tilde \tau _{iN}} = {\tau _i} - {\hat \tau _{iN}},{\tilde d_{iN}} = {\bar d_i} - {\hat d_{iN}},\left| {{d_i}} \right| \le {\bar d_i}$.

By utilizing \emph{Lemma 1-4}, the derivative of $V$ w.r.t. time is
\begin{equation}
\dot V(\rho ) \le  - {\theta _1}V(\rho ) - {\theta _2}V{(\rho )^{\frac{{1 + m}}{2}}}{\rm{ + }}{\theta _3}
\end{equation}
where ${\theta _1} > 0,{\theta _2} > 0,{\theta _3} > 0$.

According to the FnT stability theory in [17], the proof is completed.
\end{proof}
%3-B
\subsection{Globally composite neural FnT controller with single learning parameter}
To further diminish the updating weights, we define $L = \max \left\{ {{{\left\| {{l_i}} \right\|}^2}} \right\},\left( {i = 1,2, \ldots ,n} \right)$. 

Then, the virtual control variables ${\alpha _i}$ and the control input $u = {\alpha _n}$ are designed as follows (Method 1-single)
\begin{equation}\small
\begin{aligned}
{\alpha _1} &= \frac{1}{{{g_1}}}\left[ { - {k_1}{\zeta _1} + {{\dot y}_r} - {p_1}{\varphi _1}\left( {{\lambda _1}} \right) - {{\hat d}_1}\tanh \left( {\frac{{{\lambda _1}}}{{{\eta _{1d}}}}} \right)} \right]\\
 &- \frac{1}{{{g_1}}}\left[ {\frac{{{w_1}{\lambda _1}}}{{2a_1^2}}{{\hat L}}\psi _1^T{\psi _1} + \left( {1 - {w_1}} \right){{\hat \tau }_1}{H_1}\tanh \left( {\frac{{{H_1}{\lambda _1}}}{{{\eta _1}}}} \right)} \right]\\
{\alpha _i} &= \frac{1}{{{g_i}}}\left[ { - {k_i}{\zeta _i} + {{\dot \rho }_{i,c}} - {g_{i - 1}}{\zeta _{i - 1}} - {p_i}{\varphi _i}\left( {{\lambda _i}} \right) - {{\hat d}_i}\tanh \left( {\frac{{{\lambda _i}}}{{{\eta _{id}}}}} \right)} \right]\\
 &- \frac{1}{{{g_i}}}\left[ {\frac{{{w_i}{\lambda _i}}}{{2a_i^2}}{{\hat L}}\psi _i^T{\psi _i} + \left( {1 - {w_i}} \right){{\hat \tau }_i}{H_i}\tanh \left( {\frac{{{H_i}{\lambda _i}}}{{{\eta _i}}}} \right)} \right]
\end{aligned}
\end{equation} 
where ${\hat L},{\hat \tau _i},{\hat d_i}$ are the approximation of ${L},{\tau _i},{d_i}$, respectively, which are formulated as follows
\begin{equation}
\begin{aligned}
&\dot {\hat L} = {\beta _{h}}\sum\limits_{i = 1}^n {\frac{{{w_i}}}{{2a_i^2}}\left( {\lambda _i^2 + {\beta _{zi}}z_{iN}^2} \right)\psi _i^T{\psi _i}}  - {\beta _{1}}\hat L - {\beta _{2}}{\hat L^m}\\
&{\dot {\hat \tau }_i} = {\delta _{i1}}\left( {1 - {w_i}} \right){H_i}{\lambda _i}\tanh \left( {\frac{{{H_i}{\lambda _i}}}{{{\eta _i}}}} \right) - {\delta _{i2}}{{\hat \tau }_i} - {\delta _{i3}}\hat \tau _i^m\\
&{\dot {\hat d}_i} = {q_{i1}}{\lambda _i}\tanh \left( {\frac{{{\lambda _i}}}{{{\eta _{id}}}}} \right) - {q_{i2}}{{\hat d}_i} - {q_{i3}}\hat d_i^m
\end{aligned}
\end{equation}

The prediction errors ${z_{iN}} = {\rho _i} - {\hat \rho _i}$ are extracted from the following improved FnT SPEM (denote ${\rho_{n+1}}=u$):
\begin{equation}
\begin{aligned}
{\dot{\hat \rho }_i} =& \frac{{{w_i}\left( {{{\bar \rho }_i}} \right)}}{{2a_i^2}}{{\hat L}}{z_{iN}}\psi _i^T{\psi _i} + \left( {1 - {w_i}} \right){{\hat \tau }_{iN}}{H_i}\tanh \left( {\frac{{{H_i}{z_{iN}}}}{{{\eta _{iN}}}}} \right)\\
 &+ {g_i}{\rho _{i+1}} + {r_{i1}}{z_{iN}} + {r_{i2}}z_{iN}^m + {{\hat d}_{iN}}\tanh \left( {\frac{{{z_{iN}}}}{{{\eta _{idN}}}}} \right)
\end{aligned}
\end{equation}
where
\begin{equation}\small
\begin{aligned}
&{\dot {\hat \tau }_{iN}} = {\delta _{i1N}}\left( {1 - {w_i}} \right){H_i}{z_{iN}}\tanh \left( {\frac{{{H_i}{z_{iN}}}}{{{\eta _{iN}}}}} \right) - {\delta _{i2N}}{{\hat \tau }_{iN}} - {\delta _{i3N}}\hat \tau _{iN}^m\\
&{\dot {\hat d}_{iN}} = {q_{i1N}}{{z_{iN}}}\tanh \left( {\frac{{{z_{iN}}}}{{{\eta _{idN}}}}} \right) - {q_{i2N}}{{\hat d}_{iN}} - {q_{i3N}}\hat d_{iN}^m
\end{aligned}
\end{equation}

%Method 2-3-single
If we define $N = \max \left\{ {\max \left\{ {{{\bar \varepsilon }_i},\left\| {{l_i}} \right\|} \right\}} \right\},\left( {i = 1,2, \ldots ,n} \right)$ (Method 2-single) or $N = \max \left\{ {\left\| {{l_i}} \right\|} \right\},\left( {i = 1,2, \ldots ,n} \right)$ (Method 3-single), then another virtual control variables ${\alpha _i}$ and the control input $u = {\alpha _n}$ are constructed  as follows
\begin{equation}\small
\begin{aligned}
&{\alpha _1} = \frac{1}{{{g_1}}}\left[ { - {k_1}{\zeta _1} + {{\dot y}_r} - {p_1}{\varphi _1}\left( {{\lambda _1}} \right) - {{\hat d}_1}\tanh \left( {\frac{{{\lambda _1}}}{{{\eta _{1d}}}}} \right)} \right]\\
 &- \frac{1}{{{g_1}}}\left[ {{w_1}{{\hat N}}{\psi _{1h}}\tanh \left( {\frac{{{\psi _{1h}}{\lambda _1}}}{{{\eta _{1\theta }}}}} \right) + \left( {1 - {w_1}} \right){{\hat \tau }_1}{H_1}\tanh \left( {\frac{{{H_1}{\lambda _1}}}{{{\eta _1}}}} \right)} \right]\\
&{\alpha _i} = \frac{1}{{{g_i}}}\left[ { - {k_i}{\zeta _i} + {{\dot \rho }_{i,c}} - {g_{i - 1}}{\zeta _{i - 1}} - {p_i}{\varphi _i}\left( {{\lambda _i}} \right) - {{\hat d}_i}\tanh \left( {\frac{{{\lambda _i}}}{{{\eta _{id}}}}} \right)} \right]\\
 &- \frac{1}{{{g_i}}}\left[ {{w_i}{{\hat N}}{\psi _{ih}}\tanh \left( {\frac{{{\psi _{ih}}{\lambda _i}}}{{{\eta _{i\theta }}}}} \right) + \left( {1 - {w_i}} \right){{\hat \tau }_i}{H_i}\tanh \left( {\frac{{{H_i}{\lambda _i}}}{{{\eta _i}}}} \right)} \right]
\end{aligned}
\end{equation}
where ${\hat N},{\hat \tau _i},{\hat d_i}$ are the approximation of ${N},{\tau _i},{d_i}$ respectively, which are expressed as follows
\begin{equation}\small
\begin{aligned}
\dot {\hat N} = &{\beta _{h}}\sum\limits_{i = 1}^n {{w_i}\left[ {{\lambda _i}{\psi _{ih}}\tanh \left( {\frac{{{\lambda _i}{\psi _{ih}}}}{{{\eta _{i\theta }}}}} \right) + {\beta _{zi}}{z_{iN}}{\psi _{ih}}\tanh \left( {\frac{{{z_{iN}}{\psi _{ih}}}}{{{\eta _{i\theta N}}}}} \right)} \right]}  \\
&- {\beta _{1}}\hat N - {\beta _{2}}{\hat N^m}\\
{\dot {\hat \tau }_i} = &{\delta _{i1}}\left( {1 - {w_i}} \right){H_i}{\lambda _i}\tanh \left( {\frac{{{H_i}{\lambda _i}}}{{{\eta _i}}}} \right) - {\delta _{i2}}{{\hat \tau }_i} - {\delta _{i3}}\hat \tau _i^m\\
{\dot {\hat d}_i} =& {q_{i1}}{\lambda _i}\tanh \left( {\frac{{{\lambda _i}}}{{{\eta _{id}}}}} \right) - {q_{i2}}{{\hat d}_i} - {q_{i3}}\hat d_i^m
\end{aligned}
\end{equation}

The prediction errors ${z_{iN}} = {\rho _i} - {\hat \rho _i}$ are extracted from the following improved FnT SPEM (denote ${\rho_{n+1}}=u$):
\begin{equation}
\begin{aligned}
{\dot{\hat \rho }_i} = &{w_i}{{\hat N}}{\psi _{ih}}\tanh \left( {\frac{{{\psi _{ih}}{z_{iN}}}}{{{\eta _{i\theta N}}}}} \right) + \left( {1 - {w_i}} \right){{\hat \tau }_{iN}}{H_i}\tanh \left( {\frac{{{H_i}{z_{iN}}}}{{{\eta _{iN}}}}} \right)\\
& + {g_i}{\rho _{i+1}} + {r_{i1}}{z_{iN}} + {r_{i2}}z_{iN}^m + {{\hat d}_{iN}}\tanh \left( {\frac{{{z_{iN}}}}{{{\eta _{idN}}}}} \right)
\end{aligned}
\end{equation}
where
\begin{equation}\small
\begin{aligned}
&{\dot {\hat \tau }_{iN}} = {\delta _{i1N}}\left( {1 - {w_i}} \right){H_i}{z_{iN}}\tanh \left( {\frac{{{H_i}{z_{iN}}}}{{{\eta _{iN}}}}} \right) - {\delta _{i2N}}{{\hat \tau }_{iN}} - {\delta _{i3N}}\hat \tau _{iN}^m\\
&{\dot {\hat d}_{iN}} = {q_{i1N}}{{z_{iN}}}\tanh \left( {\frac{{{z_{iN}}}}{{{\eta _{idN}}}}} \right) - {q_{i2N}}{{\hat d}_{iN}} - {q_{i3N}}\hat d_{iN}^m
\end{aligned}
\end{equation}
\begin{theorem}
Considering the uncertain strict-feedback system (1), the compensation system (14), the prediction errors and modified finite-time SPEM (29) or (33), the virtual control law and controller (27) or (31), the adaptive laws (28) or (32), all signals in the closed-loop system are globally FnT bounded, and the tracking error tends to an arbitrarily small domain around zero for finite time. 
\end{theorem}
\begin{proof}
For controller (27), the following Lyapunov function is developed
\begin{equation}
\begin{aligned}
V = &\frac{1}{2}\sum\limits_{i = 1}^n {\left( {\lambda _i^2 + {\beta _{zi}}z_{iN}^2 + \sigma _i^2} \right)}+ \frac{{\tilde L^2}}{{2{\beta _{h}}}} \\
& + \frac{1}{2}\sum\limits_{i = 1}^n {\left( {\delta _{i1}^{ - 1}\tilde \tau _i^2 + \delta _{i1N}^{ - 1}\tilde \tau _{iN}^2 + q_{i1}^{ - 1}\tilde d_i^2 + q_{i1N}^{ - 1}\tilde d_{iN}^2} \right)}, 
\end{aligned}
\end{equation}
and for controller (31), the Lyapunov function is designed as below
\begin{equation}
\begin{aligned}
V = &\frac{1}{2}\sum\limits_{i = 1}^n {\left( {\lambda _i^2  + {\beta _{zi}}z_{iN}^2 + \sigma _i^2} \right)}+ \frac{{\tilde N^2}}{{{\beta _{h}}}} \\
& + \frac{1}{2}\sum\limits_{i = 1}^n {\left( {\delta _{i1}^{ - 1}\tilde \tau _i^2 + \delta _{i1N}^{ - 1}\tilde \tau _{iN}^2 + q_{i1}^{ - 1}\tilde d_i^2 + q_{i1N}^{ - 1}\tilde d_{iN}^2} \right)}, 
\end{aligned}
\end{equation}
where ${\tilde L} = {L} - {\hat L},{\tilde N} = {N} - {\hat N},{\tilde \tau _i} = {\tau _i} - {\hat \tau _i},{\tilde d_i} = {\bar d_i} - {\hat d_i},{\tilde \tau _{iN}} = {\tau _i} - {\hat \tau _{iN}},{\tilde d_{iN}} = {\bar d_i} - {\hat d_{iN}},\left| {{d_i}} \right| \le {\bar d_i}$.

By utilizing \emph{Lemma 1-4}, the derivative of $V$ w.r.t. time is
\begin{equation}
\dot V(\rho ) \le  - {\vartheta _1}V(\rho ) - {\vartheta _2}V{(\rho )^{\frac{{1 + m}}{2}}}{\rm{ + }}{\vartheta _3}
\end{equation}
where ${\vartheta _1} > 0,{\vartheta _2} > 0,{\vartheta _3} > 0$.

In terms of the FnT stability theory in [17], the proof is completed.
\end{proof}
%3-C
\subsection{Globally neural FxT controller with composite learning laws}
Define tracking errors as follows
\begin{equation}
\begin{aligned}
&{\zeta _1} = {\rho _1} - {y_r}\\
&{\zeta _i} = {\rho _i} - {\rho _{i,c}},\quad i = 2,3, \ldots ,n
\end{aligned}
\end{equation}
where ${\rho _{i,c}}$ are obtained by the FxT command filter \cite{2011Uniform} with virtual control variables ${\alpha _{i - 1}}$ being inputs.

To weaken the impact of $\left( {{\rho _{i,c}} - {\alpha _{i - 1}}} \right)$, the following compensation system is introduced [25]
\begin{equation}
\begin{aligned}
&{{\dot \sigma }_1} =  - {k_{11}}\sigma _1^m - {k_{12}}\sigma _1^r + {g_1}\left( {{\rho _{2,c}} - {\alpha _1}} \right) + {g_1}{\sigma _2}\\
&{{\dot \sigma }_i} =  - {k_{i1}}\sigma _i^m - {k_{i2}}\sigma _i^r + {g_i}\left( {{\rho _{\left( {i + 1} \right),c}} - {\alpha _i}} \right) - {g_{i - 1}}{\sigma _{i - 1}} + {g_i}{\sigma _{i + 1}}\\
&{{\dot \sigma }_n} =  - {k_{n1}}\sigma _n^m - {k_{n2}}\sigma _n^r - {g_{n - 1}}{\sigma _{n - 1}}
\end{aligned}
\end{equation}
where ${\sigma _i}\left( 0 \right) = 0$, and ${k_{i1}} > 0,{k_{i2}} > 0,r = {r_2}/{r_1} > 1$ are designed parameters with $r_1, r_2$ being odd numbers.

Then, the virtual control variables ${\alpha _i}$ and the control input $u = {\alpha _n}$ are developed as follows (Method 4)
\begin{equation}\small
\begin{aligned}
{\alpha _1} &= \frac{1}{{{g_1}}}\left[ { - {k_{12}}\lambda _1^{r} + {{\dot y}_r} - {k_{11}}{\varphi _1}\left( {{\lambda _1}} \right) - {{\hat d}_1}\tanh \left( {\frac{{{\lambda _1}}}{{{\eta _{1d}}}}} \right)} \right]\\
 &- \frac{1}{{{g_1}}}\left[ {\frac{{{w_1}\left( {{{\bar \rho }_1}} \right){\lambda _1}}}{{2a_1^2}}{{\hat L}_1}\psi _1^T{\psi _1} + \left( {1 - {w_1}} \right){{\hat \tau }_1}{H_1}\tanh \left( {\frac{{{H_1}{\lambda _1}}}{{{\eta _1}}}} \right)} \right]\\
{\alpha _i} &= \frac{1}{{{g_i}}}\left[ { - {k_{i2}}\lambda _i^{r} + {{\dot \rho }_{i,c}} - {g_{i - 1}}{\zeta _{i - 1}} - {k_{i1}}{\varphi _i}\left( {{\lambda _i}} \right) - {{\hat d}_i}\tanh \left( {\frac{{{\lambda _i}}}{{{\eta _{id}}}}} \right)} \right]\\
 &- \frac{1}{{{g_i}}}\left[ {\frac{{{w_i}\left( {{{\bar \rho }_i}} \right){\lambda _i}}}{{2a_i^2}}{{\hat L}_i}\psi _i^T{\psi _i} + \left( {1 - {w_i}} \right){{\hat \tau }_i}{H_i}\tanh \left( {\frac{{{H_i}{\lambda _i}}}{{{\eta _i}}}} \right)} \right]
\end{aligned}
\end{equation} 
 where ${a_i} > 0,{\eta _i} > 0,{\eta _{id}} > 0,{\lambda _i} = {\zeta _i} - {\sigma _i}$, and ${\hat L_i},{\hat \tau _i},{\hat d_i}$ are the estimation of ${L_i},{\tau _i},{d_i}$, respectively, which are as follows
\begin{equation}
\begin{aligned}
&{\dot{ \hat L}_i} = {\beta _{hi}}{w_i}\left( {\lambda _i^2 + {\beta _{zi}}z_{iN}^2} \right)\psi _i^T{\psi _i} - {\beta _{i1}}{{\hat L}_i}^r - {\beta _{i2}}\hat L_i^m\\
&{\dot {\hat \tau }_i} = {\delta _{i1}}\left( {1 - {w_i}} \right){H_i}{\lambda _i}\tanh \left( {\frac{{{H_i}{\lambda _i}}}{{{\eta _i}}}} \right) - {\delta _{i2}}{{\hat \tau }_i}^r - {\delta _{i3}}\hat \tau _i^m\\
&{\dot {\hat d}_i} = {q_{i1}}{\lambda _i}\tanh \left( {\frac{{{\lambda _i}}}{{{\eta _{id}}}}} \right) - {q_{i2}}{{\hat d}_i}^r - {q_{i3}}\hat d_i^m
\end{aligned}
\end{equation}
with ${\beta _{hi}},{\beta _{zi}},{\beta _{i1}},{\beta _{i2}},{\delta _{i1}},{\delta _{i2}},{\delta _{i3}},{q_{i1}},{q_{i2}},{q_{i3}}$ being positive design parameters.

The prediction errors ${z_{iN}} = {\rho _i} - {\hat \rho _i}$ are extracted from the following modified FxT SPEM (denote ${\rho_{n+1}}=u$):
\begin{equation}
\begin{aligned}
{\dot{\hat \rho }_i} =& \frac{{{w_i}\left( {{{\bar \rho }_i}} \right)}}{{2a_i^2}}{{\hat L}_i}{z_{iN}}\psi _i^T{\psi _i} + \left( {1 - {w_i}} \right){{\hat \tau }_{iN}}{H_i}\tanh \left( {\frac{{{H_i}{z_{iN}}}}{{{\eta _{iN}}}}} \right)\\
 &+ {g_i}{\rho _{i+1}} + {r_{i1}}z_{iN}^r + {r_{i2}}z_{iN}^m + {{\hat d}_{iN}}\tanh \left( {\frac{{{z_{iN}}}}{{{\eta _{idN}}}}} \right)
\end{aligned}
\end{equation}
where
\begin{equation}\small
\begin{aligned}
&{\dot {\hat \tau }_{iN}} = {\delta _{i1N}}\left( {1 - {w_i}} \right){H_i}{z_{iN}}\tanh \left( {\frac{{{H_i}{z_{iN}}}}{{{\eta _{iN}}}}} \right) - {\delta _{i2N}}{{\hat \tau }_{iN}}^r - {\delta _{i3N}}\hat \tau _{iN}^m\\
&{\dot {\hat d}_{iN}} = {q_{i1N}}{z_{iN}}\tanh \left( {\frac{{{z_{iN}}}}{{{\eta _{idN}}}}} \right) - {q_{i2N}}{{\hat d}_{iN}}^r - {q_{i3N}}\hat d_{iN}^m
\end{aligned}
\end{equation}
with ${\eta _{iN}},{\eta _{idN}},{\delta _{i1N}},{\delta _{i2N}},{\delta _{i3N}},{q_{i1N}},{q_{i2N}},{q_{i3N}}$ being positive design parameters.

%Method 5-6
If we define ${N_i} = \max \left\{ {{{\bar \varepsilon }_i},\left\| {{l_i}} \right\|} \right\},{\psi _{ih}} = \left\| {{\psi _i}} \right\| + 1$ (Method 5) or ${N_i} = \left\| {{l_i}} \right\|,{\psi _{ih}} = \left\| {{\psi _i}} \right\|$ (Method 6), then another virtual control variables ${\alpha _i}$ and the control input $u = {\alpha _n}$ are constructed as follows
\begin{equation}\small
\begin{aligned}
&{\alpha _1} = \frac{1}{{{g_1}}}\left[ { - {k_{12}}\lambda _1^{r} + {{\dot y}_r} - {k_{11}}{\varphi _1}\left( {{\lambda _1}} \right) - {{\hat d}_1}\tanh \left( {\frac{{{\lambda _1}}}{{{\eta _{1d}}}}} \right)} \right]\\
 &- \frac{1}{{{g_1}}}\left[ {{w_1}{{\hat N}_1}{\psi _{1h}}\tanh \left( {\frac{{{\psi _{1h}}{\lambda _1}}}{{{\eta _{1\theta }}}}} \right) + \left( {1 - {w_1}} \right){{\hat \tau }_1}{H_1}\tanh \left( {\frac{{{H_1}{\lambda _1}}}{{{\eta _1}}}} \right)} \right]\\
&{\alpha _i} = \frac{1}{{{g_i}}}\left[ { - {k_{i2}}\lambda _i^{r} + {{\dot \rho }_{i,c}} - {g_{i - 1}}{\zeta _{i - 1}} - {k_{i1}}{\varphi _i}\left( {{\lambda _i}} \right) - {{\hat d}_i}\tanh \left( {\frac{{{\lambda _i}}}{{{\eta _{id}}}}} \right)} \right]\\
 &- \frac{1}{{{g_i}}}\left[ {{w_i}{{\hat N}_i}{\psi _{ih}}\tanh \left( {\frac{{{\psi _{ih}}{\lambda _i}}}{{{\eta _{i\theta }}}}} \right) + \left( {1 - {w_i}} \right){{\hat \tau }_i}{H_i}\tanh \left( {\frac{{{H_i}{\lambda _i}}}{{{\eta _i}}}} \right)} \right]
\end{aligned}
\end{equation}
where ${\hat N_i},{\hat \tau _i},{\hat d_i}$ are the estimation of ${N_i},{\tau _i},{d_i}$ respectively, which are as follows
\begin{equation}\small
\begin{aligned}
{\dot {\hat N}_i} = &{\beta _{hi}}{w_i}\left[ {{\lambda _i}{\psi _{ih}}\tanh \left( {\frac{{{\lambda _i}{\psi _{ih}}}}{{{\eta _{i\theta }}}}} \right) + {\beta _{zi}}{z_{iN}}{\psi _{ih}}\tanh \left( {\frac{{{z_{iN}}{\psi _{ih}}}}{{{\eta _{i\theta N}}}}} \right)} \right]\\
 &- {\beta _{i1}}{{\hat N}_i}^r - {\beta _{i2}}\hat N_i^m\\
{\dot {\hat \tau }_i} = &{\delta _{i1}}\left( {1 - {w_i}} \right){H_i}{\lambda _i}\tanh \left( {\frac{{{H_i}{\lambda _i}}}{{{\eta _i}}}} \right) - {\delta _{i2}}{{\hat \tau }_i}^r - {\delta _{i3}}\hat \tau _i^m\\
{\dot {\hat d}_i} =& {q_{i1}}{\lambda _i}\tanh \left( {\frac{{{\lambda _i}}}{{{\eta _{id}}}}} \right) - {q_{i2}}{{\hat d}_i}^r - {q_{i3}}\hat d_i^m
\end{aligned}
\end{equation}
with ${\eta _{1\theta }} > 0,{\eta _{1\theta N}} > 0$ being design parameters.

The prediction errors ${z_{iN}} = {\rho _i} - {\hat \rho _i}$ are extracted from the following modified FxT SPEM (denote ${\rho_{n+1}}=u$):
\begin{equation}
\begin{aligned}
{\dot{\hat \rho }_i} = &{w_i}{{\hat N}_i}{\psi _{ih}}\tanh \left( {\frac{{{\psi _{ih}}{z_{iN}}}}{{{\eta _{i\theta N}}}}} \right) + \left( {1 - {w_i}} \right){{\hat \tau }_{iN}}{H_i}\tanh \left( {\frac{{{H_i}{z_{iN}}}}{{{\eta _{iN}}}}} \right)\\
& + {g_i}{\rho _{i+1}} + {r_{i1}}z_{iN}^r + {r_{i2}}z_{iN}^m + {{\hat d}_{iN}}\tanh \left( {\frac{{{z_{iN}}}}{{{\eta _{idN}}}}} \right)
\end{aligned}
\end{equation}
where
\begin{equation}\small
\begin{aligned}
&{\dot {\hat \tau }_{iN}} = {\delta _{i1N}}\left( {1 - {w_i}} \right){H_i}{z_{iN}}\tanh \left( {\frac{{{H_i}{z_{iN}}}}{{{\eta _{iN}}}}} \right) - {\delta _{i2N}}{{\hat \tau }_{iN}}^r - {\delta _{i3N}}\hat \tau _{iN}^m\\
&{\dot {\hat d}_{iN}} = {q_{i1N}}{z_{iN}}\tanh \left( {\frac{{{z_{iN}}}}{{{\eta _{idN}}}}} \right) - {q_{i2N}}{{\hat d}_{iN}}^r - {q_{i3N}}\hat d_{iN}^m
\end{aligned}
\end{equation}
\begin{theorem}
Considering the uncertain strict-feedback system (1), the compensation system (39), the prediction errors and modified FxT SPEM (42) or (46), the virtual control law and controller (40) or (44), the adaptive laws (41) or (45), all signals in the closed-loop system are globally FxT bounded, and the tracking error tends to an arbitrarily small domain near zero for fixed time. 
\end{theorem}
\begin{proof}
For controller (40), the Lyapunov function is constructed as below
\begin{equation}
\begin{aligned}
V = &\frac{1}{2}\sum\limits_{i = 1}^n {\left( {\lambda _i^2 + \frac{{\tilde L_i^2}}{{2a_i^2{\beta _{hi}}}} + {\beta _{zi}}z_{iN}^2 + \sigma _i^2} \right)} \\
& + \frac{1}{2}\sum\limits_{i = 1}^n {\left( {\delta _{i1}^{ - 1}\tilde \tau _i^2 + \delta _{i1N}^{ - 1}\tilde \tau _{iN}^2 + q_{i1}^{ - 1}\tilde d_i^2 + q_{i1N}^{ - 1}\tilde d_{iN}^2} \right)}, 
\end{aligned}
\end{equation}
and for controller (44), the following Lyapunov function is designed:
\begin{equation}
\begin{aligned}
V = &\frac{1}{2}\sum\limits_{i = 1}^n {\left( {\lambda _i^2 + \frac{{\tilde N_i^2}}{{{\beta _{hi}}}} + {\beta _{zi}}z_{iN}^2 + \sigma _i^2} \right)} \\
& + \frac{1}{2}\sum\limits_{i = 1}^n {\left( {\delta _{i1}^{ - 1}\tilde \tau _i^2 + \delta _{i1N}^{ - 1}\tilde \tau _{iN}^2 + q_{i1}^{ - 1}\tilde d_i^2 + q_{i1N}^{ - 1}\tilde d_{iN}^2} \right)}, 
\end{aligned}
\end{equation}
where ${\tilde L_i} = {L_i} - {\hat L_i},{\tilde N_i} = {N_i} - {\hat N_i},{\tilde \tau _i} = {\tau _i} - {\hat \tau _i},{\tilde d_i} = {\bar d_i} - {\hat d_i},{\tilde \tau _{iN}} = {\tau _i} - {\hat \tau _{iN}},{\tilde d_{iN}} = {\bar d_i} - {\hat d_{iN}},\left| {{d_i}} \right| \le {\bar d_i}$.

By utilizing \emph{Lemma 1-4}, the derivative of $V$ w.r.t. time is
\begin{equation}
\dot V(\rho ) \le  - {\bar \theta _1}V(\rho )^{\frac{{1 + r}}{2}} - {\bar \theta _2}V{(\rho )^{\frac{{1 + m}}{2}}}{\rm{ + }}{\bar \theta _3}
\end{equation}
where ${\bar \theta _1} > 0,{\bar \theta _2} > 0,{\bar \theta _3} > 0$.

According to the FxT stability theory in \cite{fxt1,fxt_own}, the proof is completed.
\end{proof}
%3-D
\subsection{Globally composite neural FxT controller with single learning parameter}
To further decline the updating weights, we define $L = \max \left\{ {{{\left\| {{l_i}} \right\|}^2}} \right\},\left( {i = 1,2, \ldots ,n} \right)$. 

Then, the virtual control variables ${\alpha _i}$ and the control input $u = {\alpha _n}$ are developed as follows (Method 4-single)
\begin{equation}\small
\begin{aligned}
{\alpha _1} &= \frac{1}{{{g_1}}}\left[ { - {k_{12}}\lambda _1^{r} + {{\dot y}_r} - {k_{11}}{\varphi _1}\left( {{\lambda _1}} \right) - {{\hat d}_1}\tanh \left( {\frac{{{\lambda _1}}}{{{\eta _{1d}}}}} \right)} \right]\\
 &- \frac{1}{{{g_1}}}\left[ {\frac{{{w_1}{\lambda _1}}}{{2a_1^2}}{{\hat L}}\psi _1^T{\psi _1} + \left( {1 - {w_1}} \right){{\hat \tau }_1}{H_1}\tanh \left( {\frac{{{H_1}{\lambda _1}}}{{{\eta _1}}}} \right)} \right]\\
{\alpha _i} &= \frac{1}{{{g_i}}}\left[ { - {k_{i2}}\lambda _i^{r} + {{\dot \rho }_{i,c}} - {g_{i - 1}}{\zeta _{i - 1}} - {k_{i1}}{\varphi _i}\left( {{\lambda _i}} \right) - {{\hat d}_i}\tanh \left( {\frac{{{\lambda _i}}}{{{\eta _{id}}}}} \right)} \right]\\
 &- \frac{1}{{{g_i}}}\left[ {\frac{{{w_i}{\lambda _i}}}{{2a_i^2}}{{\hat L}}\psi _i^T{\psi _i} + \left( {1 - {w_i}} \right){{\hat \tau }_i}{H_i}\tanh \left( {\frac{{{H_i}{\lambda _i}}}{{{\eta _i}}}} \right)} \right]
\end{aligned}
\end{equation} 
where ${\hat L},{\hat \tau _i},{\hat d_i}$ are the estimation of ${L},{\tau _i},{d_i}$, respectively, which are as follows
\begin{equation}
\begin{aligned}
&\dot {\hat L} = {\beta _{h}}\sum\limits_{i = 1}^n {\frac{{{w_i}}}{{2a_i^2}}\left( {\lambda _i^2 + {\beta _{zi}}z_{iN}^2} \right)\psi _i^T{\psi _i}}  - {\beta _{1}}{\hat L^r} - {\beta _{2}}{\hat L^m}\\
&{\dot {\hat \tau }_i} = {\delta _{i1}}\left( {1 - {w_i}} \right){H_i}{\lambda _i}\tanh \left( {\frac{{{H_i}{\lambda _i}}}{{{\eta _i}}}} \right) - {\delta _{i2}}{{\hat \tau }_i}^r - {\delta _{i3}}\hat \tau _i^m\\
&{\dot {\hat d}_i} = {q_{i1}}{\lambda _i}\tanh \left( {\frac{{{\lambda _i}}}{{{\eta _{id}}}}} \right) - {q_{i2}}{{\hat d}_i}^r - {q_{i3}}\hat d_i^m
\end{aligned}
\end{equation}

The prediction errors ${z_{iN}} = {\rho _i} - {\hat \rho _i}$ are extracted from the following improved fixed-time SPEM (denote ${\rho_{n+1}}=u$):
\begin{equation}
\begin{aligned}
{\dot{\hat \rho }_i} =& \frac{{{w_i}\left( {{{\bar \rho }_i}} \right)}}{{2a_i^2}}{{\hat L}}{z_{iN}}\psi _i^T{\psi _i} + \left( {1 - {w_i}} \right){{\hat \tau }_{iN}}{H_i}\tanh \left( {\frac{{{H_i}{z_{iN}}}}{{{\eta _{iN}}}}} \right)\\
 &+ {g_i}{\rho _{i+1}} + {r_{i1}}z_{iN}^r + {r_{i2}}z_{iN}^m + {{\hat d}_{iN}}\tanh \left( {\frac{{{z_{iN}}}}{{{\eta _{idN}}}}} \right)
\end{aligned}
\end{equation}
where
\begin{equation}\small
\begin{aligned}
&{\dot {\hat \tau }_{iN}} = {\delta _{i1N}}\left( {1 - {w_i}} \right){H_i}{z_{iN}}\tanh \left( {\frac{{{H_i}{z_{iN}}}}{{{\eta _{iN}}}}} \right) - {\delta _{i2N}}{{\hat \tau }_{iN}}^r - {\delta _{i3N}}\hat \tau _{iN}^m\\
&{\dot {\hat d}_{iN}} = {q_{i1N}}{z_{iN}}\tanh \left( {\frac{{{z_{iN}}}}{{{\eta _{idN}}}}} \right) - {q_{i2N}}{{\hat d}_{iN}}^r - {q_{i3N}}\hat d_{iN}^m
\end{aligned}
\end{equation}

%Method 5-6-single
If we define $N = \max \left\{ {\max \left\{ {{{\bar \varepsilon }_i},\left\| {{l_i}} \right\|} \right\}} \right\},\left( {i = 1,2, \ldots ,n} \right)$ (Method 5-single) or $N = \max \left\{ {\left\| {{l_i}} \right\|} \right\},\left( {i = 1,2, \ldots ,n} \right)$ (Method 6-single), then another virtual control variables ${\alpha _i}$ and the control input $u = {\alpha _n}$ are constructed as follows
\begin{equation}\small
\begin{aligned}
&{\alpha _1} = \frac{1}{{{g_1}}}\left[ { - {k_{12}}\lambda _1^{r} + {{\dot y}_r} - {k_{11}}{\varphi _1}\left( {{\lambda _1}} \right) - {{\hat d}_1}\tanh \left( {\frac{{{\lambda _1}}}{{{\eta _{1d}}}}} \right)} \right]\\
 &- \frac{1}{{{g_1}}}\left[ {{w_1}{{\hat N}}{\psi _{1h}}\tanh \left( {\frac{{{\psi _{1h}}{\lambda _1}}}{{{\eta _{1\theta }}}}} \right) + \left( {1 - {w_1}} \right){{\hat \tau }_1}{H_1}\tanh \left( {\frac{{{H_1}{\lambda _1}}}{{{\eta _1}}}} \right)} \right]\\
&{\alpha _i} = \frac{1}{{{g_i}}}\left[ { - {k_{i2}}\lambda _i^{r} + {{\dot \rho }_{i,c}} - {g_{i - 1}}{\zeta _{i - 1}} - {k_{i1}}{\varphi _i}\left( {{\lambda _i}} \right) - {{\hat d}_i}\tanh \left( {\frac{{{\lambda _i}}}{{{\eta _{id}}}}} \right)} \right]\\
 &- \frac{1}{{{g_i}}}\left[ {{w_i}{{\hat N}}{\psi _{ih}}\tanh \left( {\frac{{{\psi _{ih}}{\lambda _i}}}{{{\eta _{i\theta }}}}} \right) + \left( {1 - {w_i}} \right){{\hat \tau }_i}{H_i}\tanh \left( {\frac{{{H_i}{\lambda _i}}}{{{\eta _i}}}} \right)} \right]
\end{aligned}
\end{equation}
where ${\hat N},{\hat \tau _i},{\hat d_i}$ are the estimation of ${N},{\tau _i},{d_i}$ respectively, which are as follows
\begin{equation}\small
\begin{aligned}
\dot {\hat N} = &{\beta _{h}}\sum\limits_{i = 1}^n {{w_i}\left[ {{\lambda _i}{\psi _{ih}}\tanh \left( {\frac{{{\lambda _i}{\psi _{ih}}}}{{{\eta _{i\theta }}}}} \right) + {\beta _{zi}}{z_{iN}}{\psi _{ih}}\tanh \left( {\frac{{{z_{iN}}{\psi _{ih}}}}{{{\eta _{i\theta N}}}}} \right)} \right]}  \\
&- {\beta _{1}}{\hat N^r} - {\beta _{2}}{\hat N^m}\\
{\dot {\hat \tau }_i} = &{\delta _{i1}}\left( {1 - {w_i}} \right){H_i}{\lambda _i}\tanh \left( {\frac{{{H_i}{\lambda _i}}}{{{\eta _i}}}} \right) - {\delta _{i2}}{{\hat \tau }_i}^r - {\delta _{i3}}\hat \tau _i^m\\
{\dot {\hat d}_i} =& {q_{i1}}{\lambda _i}\tanh \left( {\frac{{{\lambda _i}}}{{{\eta _{id}}}}} \right) - {q_{i2}}{{\hat d}_i}^r - {q_{i3}}\hat d_i^m
\end{aligned}
\end{equation}

The prediction errors ${z_{iN}} = {\rho _i} - {\hat \rho _i}$ are extracted from the following improved fixed-time SPEM (denote ${\rho_{n+1}}=u$):
\begin{equation}
\begin{aligned}
{\dot{\hat \rho }_i} = &{w_i}{{\hat N}}{\psi _{ih}}\tanh \left( {\frac{{{\psi _{ih}}{z_{iN}}}}{{{\eta _{i\theta N}}}}} \right) + \left( {1 - {w_i}} \right){{\hat \tau }_{iN}}{H_i}\tanh \left( {\frac{{{H_i}{z_{iN}}}}{{{\eta _{iN}}}}} \right)\\
& + {g_i}{\rho _{i+1}} + {r_{i1}}z_{iN}^r + {r_{i2}}z_{iN}^m + {{\hat d}_{iN}}\tanh \left( {\frac{{{z_{iN}}}}{{{\eta _{idN}}}}} \right)
\end{aligned}
\end{equation}
where
\begin{equation}\small
\begin{aligned}
&{\dot {\hat \tau }_{iN}} = {\delta _{i1N}}\left( {1 - {w_i}} \right){H_i}{z_{iN}}\tanh \left( {\frac{{{H_i}{z_{iN}}}}{{{\eta _{iN}}}}} \right) - {\delta _{i2N}}{{\hat \tau }_{iN}}^r - {\delta _{i3N}}\hat \tau _{iN}^m\\
&{\dot {\hat d}_{iN}} = {q_{i1N}}{z_{iN}}\tanh \left( {\frac{{{z_{iN}}}}{{{\eta _{idN}}}}} \right) - {q_{i2N}}{{\hat d}_{iN}}^r - {q_{i3N}}\hat d_{iN}^m
\end{aligned}
\end{equation}
\begin{theorem}
Considering the uncertain strict-feedback system (1), the compensation system (39), the prediction errors and modified FxT SPEM (53) or (57), the virtual control law and controller (51) or (55), the adaptive laws (52) or (56), all signals in the closed-loop system are globally FxT bounded, and the tracking error tends to an arbitrarily small domain near zero for fixed time. 
\end{theorem}
\begin{proof}
For controller (51), the Lyapunov function is constructed as below
\begin{equation}
\begin{aligned}
V = &\frac{1}{2}\sum\limits_{i = 1}^n {\left( {\lambda _i^2  + {\beta _{zi}}z_{iN}^2 + \sigma _i^2} \right)}+ \frac{{\tilde L^2}}{{2{\beta _{h}}}} \\
& + \frac{1}{2}\sum\limits_{i = 1}^n {\left( {\delta _{i1}^{ - 1}\tilde \tau _i^2 + \delta _{i1N}^{ - 1}\tilde \tau _{iN}^2 + q_{i1}^{ - 1}\tilde d_i^2 + q_{i1N}^{ - 1}\tilde d_{iN}^2} \right)}, 
\end{aligned}
\end{equation}
and for controller (55), the following Lyapunov function is designed:
\begin{equation}
\begin{aligned}
V = &\frac{1}{2}\sum\limits_{i = 1}^n {\left( {\lambda _i^2 + {\beta _{zi}}z_{iN}^2 + \sigma _i^2} \right)} + \frac{{\tilde N^2}}{{{\beta _{h}}}} \\
& + \frac{1}{2}\sum\limits_{i = 1}^n {\left( {\delta _{i1}^{ - 1}\tilde \tau _i^2 + \delta _{i1N}^{ - 1}\tilde \tau _{iN}^2 + q_{i1}^{ - 1}\tilde d_i^2 + q_{i1N}^{ - 1}\tilde d_{iN}^2} \right)}, 
\end{aligned}
\end{equation}
where ${\tilde L} = {L} - {\hat L},{\tilde N} = {N} - {\hat N},{\tilde \tau _i} = {\tau _i} - {\hat \tau _i},{\tilde d_i} = {\bar d_i} - {\hat d_i},{\tilde \tau _{iN}} = {\tau _i} - {\hat \tau _{iN}},{\tilde d_{iN}} = {\bar d_i} - {\hat d_{iN}},\left| {{d_i}} \right| \le {\bar d_i}$.

By utilizing \emph{Lemma 1-4}, the derivative of $V$ w.r.t. time is
\begin{equation}
\dot V(\rho ) \le  - {\bar \vartheta _1}V(\rho )^{\frac{{1 + r}}{2}} - {\bar \vartheta _2}V{(\rho )^{\frac{{1 + m}}{2}}}{\rm{ + }}{\bar \vartheta _3}
\end{equation}
where ${\bar \vartheta _1} > 0,{\bar \vartheta _2} > 0,{\bar \vartheta _3} > 0$.

According to the FxT stability theory in \cite{fxt1,fxt_own}, the proof is completed.
\end{proof}
%3-E
\subsection{Globally neural F-FxT controller with composite learning laws}
Define tracking errors below
\begin{equation}
\begin{aligned}
&{\zeta _1} = {\rho _1} - {y_r}\\
&{\zeta _i} = {\rho _i} - {\rho _{i,c}},\quad i = 2,3, \ldots ,n
\end{aligned}
\end{equation}
where ${\rho _{i,c}}$ are obtained by the FxT command filter \cite{2011Uniform} with virtual control variables ${\alpha _{i - 1}}$ being inputs.

To diminish the impact of $\left( {{\rho _{i,c}} - {\alpha _{i - 1}}} \right)$, the following compensation system is designed
\begin{equation}
\begin{aligned}
{{\dot \sigma }_1} = & - {k_{11}}\sigma _1^m - {k_{12}}\sigma _1^r + {g_1}\left( {{\rho _{2,c}} - {\alpha _1}} \right) + {g_1}{\sigma _2}- {k_{1}}\sigma _1\\
{{\dot \sigma }_i} = & - {k_{i1}}\sigma _i^m - {k_{i2}}\sigma _i^r + {g_i}\left( {{\rho _{\left( {i + 1} \right),c}} - {\alpha _i}} \right) - {g_{i - 1}}{\sigma _{i - 1}} \\
&+ {g_i}{\sigma _{i + 1}} - {k_{i}}\sigma _i\\
{{\dot \sigma }_n} = & - {k_{n1}}\sigma _n^m - {k_{n2}}\sigma _n^r - {g_{n - 1}}{\sigma _{n - 1}}- {k_{n}}\sigma _n
\end{aligned}
\end{equation}

Then, the virtual control variables ${\alpha _i}$ and the control input $u = {\alpha _n}$ are developed as follows (Method 7)
\begin{equation}\small
\begin{aligned}
{\alpha _1} &= \frac{1}{{{g_1}}}\left[ { - {k_{1}}{\zeta _1}- {k_{12}}\lambda _1^{r} + {{\dot y}_r} - {k_{11}}{\varphi _1}\left( {{\lambda _1}} \right) - {{\hat d}_1}\tanh \left( {\frac{{{\lambda _1}}}{{{\eta _{1d}}}}} \right)} \right]\\
 &- \frac{1}{{{g_1}}}\left[ {\frac{{{w_1}\left( {{{\bar \rho }_1}} \right){\lambda _1}}}{{2a_1^2}}{{\hat L}_1}\psi _1^T{\psi _1} + \left( {1 - {w_1}} \right){{\hat \tau }_1}{H_1}\tanh \left( {\frac{{{H_1}{\lambda _1}}}{{{\eta _1}}}} \right)} \right]\\
{\alpha _i} &= \frac{1}{{{g_i}}}\left[ { - {k_{i2}}\lambda _i^{r} + {{\dot \rho }_{i,c}} - {g_{i - 1}}{\zeta _{i - 1}} - {k_{i1}}{\varphi _i}\left( {{\lambda _i}} \right) - {{\hat d}_i}\tanh \left( {\frac{{{\lambda _i}}}{{{\eta _{id}}}}} \right)} \right]\\
 &- \frac{1}{{{g_i}}}\left[ {k_{i}}{\zeta _i}+{\frac{{{w_i}\left( {{{\bar \rho }_i}} \right){\lambda _i}}}{{2a_i^2}}{{\hat L}_i}\psi _i^T{\psi _i} + \left( {1 - {w_i}} \right){{\hat \tau }_i}{H_i}\tanh \left( {\frac{{{H_i}{\lambda _i}}}{{{\eta _i}}}} \right)} \right]
\end{aligned}
\end{equation} 
 where ${\hat L_i},{\hat \tau _i},{\hat d_i}$ are the estimation of ${L_i},{\tau _i},{d_i}$, respectively, which are as follows
\begin{equation}
\begin{aligned}
&{\dot{ \hat L}_i} = {\beta _{hi}}{w_i}\left( {\lambda _i^2 + {\beta _{zi}}z_{iN}^2} \right)\psi _i^T{\psi _i} - {\beta _{i1}}{{\hat L}_i}^r - {\beta _{i2}}\hat L_i^m - {\beta _{i3}}\hat L_i\\
&{\dot {\hat \tau }_i} = {\delta _{i1}}\left( {1 - {w_i}} \right){H_i}{\lambda _i}\tanh \left( {\frac{{{H_i}{\lambda _i}}}{{{\eta _i}}}} \right) - {\delta _{i2}}{{\hat \tau }_i}^r - {\delta _{i3}}\hat \tau _i^m- {\delta _{i4}}\hat \tau _i\\
&{\dot {\hat d}_i} = {q_{i1}}{\lambda _i}\tanh \left( {\frac{{{\lambda _i}}}{{{\eta _{id}}}}} \right) - {q_{i2}}{{\hat d}_i}^r - {q_{i3}}\hat d_i^m- {q_{i4}}\hat d_i
\end{aligned}
\end{equation}
with ${\beta _{hi}},{\beta _{zi}},{\beta _{i1}},{\beta _{i2}},{\beta _{i3}},{\delta _{i1}},{\delta _{i2}},{\delta _{i3}},{\delta _{i4}},{q_{i1}},{q_{i2}},{q_{i3}},{q_{i4}}$ being positive design parameters.

The prediction errors ${z_{iN}} = {\rho _i} - {\hat \rho _i}$ are extracted from the following modified F-FxT SPEM (denote ${\rho_{n+1}}=u$):
\begin{equation}
\begin{aligned}
{\dot{\hat \rho }_i} &= \frac{{{w_i}\left( {{{\bar \rho }_i}} \right)}}{{2a_i^2}}{{\hat L}_i}{z_{iN}}\psi _i^T{\psi _i} + \left( {1 - {w_i}} \right){{\hat \tau }_{iN}}{H_i}\tanh \left( {\frac{{{H_i}{z_{iN}}}}{{{\eta _{iN}}}}} \right)\\
 &+ {g_i}{\rho _{i+1}} + {r_{i1}}z_{iN}^r + {r_{i2}}z_{iN}^m +{r_{i3}}z_{iN}+ {{\hat d}_{iN}}\tanh \left( {\frac{{{z_{iN}}}}{{{\eta _{idN}}}}} \right)
\end{aligned}
\end{equation}
where
\begin{equation}
\begin{aligned}
{\dot {\hat \tau }_{iN}} =& {\delta _{i1N}}\left( {1 - {w_i}} \right){H_i}{z_{iN}}\tanh \left( {\frac{{{H_i}{z_{iN}}}}{{{\eta _{iN}}}}} \right) - {\delta _{i2N}}{{\hat \tau }_{iN}}^r \\
&- {\delta _{i3N}}\hat \tau _{iN}^m- {\delta _{i4N}}\hat \tau _{iN},\\
{\dot {\hat d}_{iN}} = &{q_{i1N}}{z_{iN}}\tanh \left( {\frac{{{z_{iN}}}}{{{\eta _{idN}}}}} \right) - {q_{i2N}}{{\hat d}_{iN}}^r - {q_{i3N}}\hat d_{iN}^m- {q_{i4N}}\hat d_{iN}
\end{aligned}
\end{equation}
with ${r_{i3}},{\delta _{i4N}},{q_{i4N}}$ being positive design parameters.

%Method 8-9
If we define ${N_i} = \max \left\{ {{{\bar \varepsilon }_i},\left\| {{l_i}} \right\|} \right\},{\psi _{ih}} = \left\| {{\psi _i}} \right\| + 1$ (Method 8) or ${N_i} = \left\| {{l_i}} \right\|,{\psi _{ih}} = \left\| {{\psi _i}} \right\|$ (Method 9), then another virtual control variables ${\alpha _i}$ and the control input $u = {\alpha _n}$ are constructed as follows
\begin{equation}\small
\begin{aligned}
&{\alpha _1} = \frac{1}{{{g_1}}}\left[ { - {k_{1}}{\zeta _1}- {k_{12}}\lambda _1^{r} + {{\dot y}_r} - {k_{11}}{\varphi _1}\left( {{\lambda _1}} \right) - {{\hat d}_1}\tanh \left( {\frac{{{\lambda _1}}}{{{\eta _{1d}}}}} \right)} \right]\\
 &- \frac{1}{{{g_1}}}\left[ {{w_1}{{\hat N}_1}{\psi _{1h}}\tanh \left( {\frac{{{\psi _{1h}}{\lambda _1}}}{{{\eta _{1\theta }}}}} \right) + \left( {1 - {w_1}} \right){{\hat \tau }_1}{H_1}\tanh \left( {\frac{{{H_1}{\lambda _1}}}{{{\eta _1}}}} \right)} \right]\\
&{\alpha _i} = \frac{1}{{{g_i}}}\left[ { - {k_{i2}}\lambda _i^{r} + {{\dot \rho }_{i,c}} - {g_{i - 1}}{\zeta _{i - 1}} - {k_{i1}}{\varphi _i}\left( {{\lambda _i}} \right) - {{\hat d}_i}\tanh \left( {\frac{{{\lambda _i}}}{{{\eta _{id}}}}} \right)} \right]\\
 &- \frac{1}{{{g_i}}}\left[ {k_{i}}{\zeta _i}+{{w_i}{{\hat N}_i}{\psi _{ih}}\tanh \left( {\frac{{{\psi _{ih}}{\lambda _i}}}{{{\eta _{i\theta }}}}} \right) + \left( {1 - {w_i}} \right){{\hat \tau }_i}{H_i}\tanh \left( {\frac{{{H_i}{\lambda _i}}}{{{\eta _i}}}} \right)} \right]
\end{aligned}
\end{equation}
where ${\hat N_i},{\hat \tau _i},{\hat d_i}$ are the estimation of ${N_i},{\tau _i},{d_i}$ respectively, which are as follows
\begin{equation}\small
\begin{aligned}
{\dot {\hat N}_i} = &{\beta _{hi}}{w_i}\left[ {{\lambda _i}{\psi _{ih}}\tanh \left( {\frac{{{\lambda _i}{\psi _{ih}}}}{{{\eta _{i\theta }}}}} \right) + {\beta _{zi}}{z_{iN}}{\psi _{ih}}\tanh \left( {\frac{{{z_{iN}}{\psi _{ih}}}}{{{\eta _{i\theta N}}}}} \right)} \right]\\
 &- {\beta _{i1}}{{\hat N}_i}^r - {\beta _{i2}}\hat N_i^m- {\beta _{i3}}\hat N_i\\
{\dot {\hat \tau }_i} = &{\delta _{i1}}\left( {1 - {w_i}} \right){H_i}{\lambda _i}\tanh \left( {\frac{{{H_i}{\lambda _i}}}{{{\eta _i}}}} \right) - {\delta _{i2}}{{\hat \tau }_i}^r - {\delta _{i3}}\hat \tau _i^m - {\delta _{i4}}\hat \tau _i\\
{\dot {\hat d}_i} =& {q_{i1}}{\lambda _i}\tanh \left( {\frac{{{\lambda _i}}}{{{\eta _{id}}}}} \right) - {q_{i2}}{{\hat d}_i}^r - {q_{i3}}\hat d_i^m- {q_{i4}}\hat d_i
\end{aligned}
\end{equation}

The prediction errors ${z_{iN}} = {\rho _i} - {\hat \rho _i}$ are extracted from the following modified F-FxT SPEM (denote ${\rho_{n+1}}=u$):
\begin{equation}\small
\begin{aligned}
{\dot{\hat \rho }_i} = &{w_i}{{\hat N}_i}{\psi _{ih}}\tanh \left( {\frac{{{\psi _{ih}}{z_{iN}}}}{{{\eta _{i\theta N}}}}} \right) + \left( {1 - {w_i}} \right){{\hat \tau }_{iN}}{H_i}\tanh \left( {\frac{{{H_i}{z_{iN}}}}{{{\eta _{iN}}}}} \right)\\
& + {g_i}{\rho _{i+1}} + {r_{i1}}z_{iN}^r + {r_{i2}}z_{iN}^m + {r_{i3}}z_{iN}+{{\hat d}_{iN}}\tanh \left( {\frac{{{z_{iN}}}}{{{\eta _{idN}}}}} \right)
\end{aligned}
\end{equation}
where
\begin{equation}\small
\begin{aligned}
{\dot {\hat \tau }_{iN}} = &{\delta _{i1N}}\left( {1 - {w_i}} \right){H_i}{z_{iN}}\tanh \left( {\frac{{{H_i}{z_{iN}}}}{{{\eta _{iN}}}}} \right) - {\delta _{i2N}}{{\hat \tau }_{iN}}^r \\
&- {\delta _{i3N}}\hat \tau _{iN}^m- {\delta _{i4N}}\hat \tau _{iN}\\
{\dot {\hat d}_{iN}} = &{q_{i1N}}{z_{iN}}\tanh \left( {\frac{{{z_{iN}}}}{{{\eta _{idN}}}}} \right) - {q_{i2N}}{{\hat d}_{iN}}^r - {q_{i3N}}\hat d_{iN}^m - {q_{i4N}}\hat d_{iN}
\end{aligned}
\end{equation}
\begin{theorem}
Considering the uncertain strict-feedback system (1), the compensation system (63), the prediction errors and modified F-FxT SPEM (66) or (70), the virtual control law and controller (64) or (68), the adaptive laws (65) or (69), all signals in the closed-loop system are globally F-FxT bounded, and the tracking error tends to an arbitrarily small domain near zero for fast fixed time. 
\end{theorem}
%3-F
\subsection{Globally composite neural F-FxT controller with single learning parameter}
To further diminish the updating weights, we define $L = \max \left\{ {{{\left\| {{l_i}} \right\|}^2}} \right\},\left( {i = 1,2, \ldots ,n} \right)$. 

Then, the virtual control variables ${\alpha _i}$ and the control input $u = {\alpha _n}$ are developed as follows (Method 7-single)
\begin{equation}\small
\begin{aligned}
{\alpha _1} &= \frac{1}{{{g_1}}}\left[ { - {k_{1}}{\zeta _1}- {k_{12}}\lambda _1^{r} + {{\dot y}_r} - {k_{11}}{\varphi _1}\left( {{\lambda _1}} \right) - {{\hat d}_1}\tanh \left( {\frac{{{\lambda _1}}}{{{\eta _{1d}}}}} \right)} \right]\\
 &- \frac{1}{{{g_1}}}\left[ {\frac{{{w_1}{\lambda _1}}}{{2a_1^2}}{{\hat L}}\psi _1^T{\psi _1} + \left( {1 - {w_1}} \right){{\hat \tau }_1}{H_1}\tanh \left( {\frac{{{H_1}{\lambda _1}}}{{{\eta _1}}}} \right)} \right]\\
{\alpha _i} &= \frac{1}{{{g_i}}}\left[ { - {k_{i2}}\lambda _i^{r} + {{\dot \rho }_{i,c}} - {g_{i - 1}}{\zeta _{i - 1}} - {k_{i1}}{\varphi _i}\left( {{\lambda _i}} \right) - {{\hat d}_i}\tanh \left( {\frac{{{\lambda _i}}}{{{\eta _{id}}}}} \right)} \right]\\
 &- \frac{1}{{{g_i}}}\left[ {k_{i}}{\zeta _i}+{\frac{{{w_i}{\lambda _i}}}{{2a_i^2}}{{\hat L}}\psi _i^T{\psi _i} + \left( {1 - {w_i}} \right){{\hat \tau }_i}{H_i}\tanh \left( {\frac{{{H_i}{\lambda _i}}}{{{\eta _i}}}} \right)} \right]
\end{aligned}
\end{equation} 
where ${\hat L},{\hat \tau _i},{\hat d_i}$ are the estimation of ${L},{\tau _i},{d_i}$, respectively, which are as follows
\begin{equation}
\begin{aligned}
\dot {\hat L} = &{\beta _{h}}\sum\limits_{i = 1}^n {\frac{{{w_i}}}{{2a_i^2}}\left( {\lambda _i^2 + {\beta _{zi}}z_{iN}^2} \right)\psi _i^T{\psi _i}}  - {\beta _{1}}{\hat L^r} - {\beta _{2}}{\hat L^m} \\
&- {\beta _{3}}{\hat L}\\
{\dot {\hat \tau }_i} = &{\delta _{i1}}\left( {1 - {w_i}} \right){H_i}{\lambda _i}\tanh \left( {\frac{{{H_i}{\lambda _i}}}{{{\eta _i}}}} \right) - {\delta _{i2}}{{\hat \tau }_i}^r - {\delta _{i3}}\hat \tau _i^m\\
&- {\delta _{i4}}\hat \tau _i\\
{\dot {\hat d}_i} =& {q_{i1}}{\lambda _i}\tanh \left( {\frac{{{\lambda _i}}}{{{\eta _{id}}}}} \right) - {q_{i2}}{{\hat d}_i}^r - {q_{i3}}\hat d_i^m- {q_{i4}}\hat d_i
\end{aligned}
\end{equation}

The prediction errors ${z_{iN}} = {\rho _i} - {\hat \rho _i}$ are extracted from the following improved F-FxT SPEM (denote ${\rho_{n+1}}=u$):
\begin{equation}
\begin{aligned}
{\dot{\hat \rho }_i} =& \frac{{{w_i}\left( {{{\bar \rho }_i}} \right)}}{{2a_i^2}}{{\hat L}}{z_{iN}}\psi _i^T{\psi _i} + \left( {1 - {w_i}} \right){{\hat \tau }_{iN}}{H_i}\tanh \left( {\frac{{{H_i}{z_{iN}}}}{{{\eta _{iN}}}}} \right)\\
 &+ {g_i}{\rho _{i+1}} + {r_{i1}}z_{iN}^r + {r_{i2}}z_{iN}^m + {r_{i3}}z_{iN}+{{\hat d}_{iN}}\tanh \left( {\frac{{{z_{iN}}}}{{{\eta _{idN}}}}} \right)
\end{aligned}
\end{equation}
where
\begin{equation}\small
\begin{aligned}
{\dot {\hat \tau }_{iN}} = &{\delta _{i1N}}\left( {1 - {w_i}} \right){H_i}{z_{iN}}\tanh \left( {\frac{{{H_i}{z_{iN}}}}{{{\eta _{iN}}}}} \right) - {\delta _{i2N}}{{\hat \tau }_{iN}}^r \\
&- {\delta _{i3N}}\hat \tau _{iN}^m- {\delta _{i4N}}\hat \tau _{iN}\\
{\dot {\hat d}_{iN}} =& {q_{i1N}}{z_{iN}}\tanh \left( {\frac{{{z_{iN}}}}{{{\eta _{idN}}}}} \right) - {q_{i2N}}{{\hat d}_{iN}}^r - {q_{i3N}}\hat d_{iN}^m- {q_{i4N}}\hat d_{iN}
\end{aligned}
\end{equation}

%Method 8-9-single
If we define $N = \max \left\{ {\max \left\{ {{{\bar \varepsilon }_i},\left\| {{l_i}} \right\|} \right\}} \right\},\left( {i = 1,2, \ldots ,n} \right)$ (Method 8-single) or $N = \max \left\{ {\left\| {{l_i}} \right\|} \right\},\left( {i = 1,2, \ldots ,n} \right)$ (Method 9-single), then another virtual control variables ${\alpha _i}$ and the control input $u = {\alpha _n}$ are constructed as follows
\begin{equation}\small
\begin{aligned}
&{\alpha _1} = \frac{1}{{{g_1}}}\left[ { - {k_{1}}{\zeta _1}- {k_{12}}\lambda _1^{r} + {{\dot y}_r} - {k_{11}}{\varphi _1}\left( {{\lambda _1}} \right) - {{\hat d}_1}\tanh \left( {\frac{{{\lambda _1}}}{{{\eta _{1d}}}}} \right)} \right]\\
 &- \frac{1}{{{g_1}}}\left[ {{w_1}{{\hat N}}{\psi _{1h}}\tanh \left( {\frac{{{\psi _{1h}}{\lambda _1}}}{{{\eta _{1\theta }}}}} \right) + \left( {1 - {w_1}} \right){{\hat \tau }_1}{H_1}\tanh \left( {\frac{{{H_1}{\lambda _1}}}{{{\eta _1}}}} \right)} \right]\\
&{\alpha _i} = \frac{1}{{{g_i}}}\left[ { - {k_{i2}}\lambda _i^{r} + {{\dot \rho }_{i,c}} - {g_{i - 1}}{\zeta _{i - 1}} - {k_{i1}}{\varphi _i}\left( {{\lambda _i}} \right) - {{\hat d}_i}\tanh \left( {\frac{{{\lambda _i}}}{{{\eta _{id}}}}} \right)} \right]\\
 &- \frac{1}{{{g_i}}}\left[{k_{i}}{\zeta _i}+ {{w_i}{{\hat N}}{\psi _{ih}}\tanh \left( {\frac{{{\psi _{ih}}{\lambda _i}}}{{{\eta _{i\theta }}}}} \right) + \left( {1 - {w_i}} \right){{\hat \tau }_i}{H_i}\tanh \left( {\frac{{{H_i}{\lambda _i}}}{{{\eta _i}}}} \right)} \right]
\end{aligned}
\end{equation}
where ${\hat N},{\hat \tau _i},{\hat d_i}$ are the estimation of ${N},{\tau _i},{d_i}$ respectively, which are as follows
\begin{equation}\small
\begin{aligned}
\dot {\hat N} = &{\beta _{h}}\sum\limits_{i = 1}^n {{w_i}\left[ {{\lambda _i}{\psi _{ih}}\tanh \left( {\frac{{{\lambda _i}{\psi _{ih}}}}{{{\eta _{i\theta }}}}} \right) + {\beta _{zi}}{z_{iN}}{\psi _{ih}}\tanh \left( {\frac{{{z_{iN}}{\psi _{ih}}}}{{{\eta _{i\theta N}}}}} \right)} \right]}  \\
&- {\beta _{1}}{\hat N^r} - {\beta _{2}}{\hat N^m} - {\beta _{3}}{\hat N}\\
{\dot {\hat \tau }_i} = &{\delta _{i1}}\left( {1 - {w_i}} \right){H_i}{\lambda _i}\tanh \left( {\frac{{{H_i}{\lambda _i}}}{{{\eta _i}}}} \right) - {\delta _{i2}}{{\hat \tau }_i}^r - {\delta _{i3}}\hat \tau _i^m- {\delta _{i4}}\hat \tau _i\\
{\dot {\hat d}_i} =& {q_{i1}}{\lambda _i}\tanh \left( {\frac{{{\lambda _i}}}{{{\eta _{id}}}}} \right) - {q_{i2}}{{\hat d}_i}^r - {q_{i3}}\hat d_i^m - {q_{i4}}\hat d_i
\end{aligned}
\end{equation}

The prediction errors ${z_{iN}} = {\rho _i} - {\hat \rho _i}$ are extracted from the following improved F-FxT SPEM (denote ${\rho_{n+1}}=u$):
\begin{equation}
\begin{aligned}
{\dot{\hat \rho }_i} = &{w_i}{{\hat N}}{\psi _{ih}}\tanh \left( {\frac{{{\psi _{ih}}{z_{iN}}}}{{{\eta _{i\theta N}}}}} \right) + \left( {1 - {w_i}} \right){{\hat \tau }_{iN}}{H_i}\tanh \left( {\frac{{{H_i}{z_{iN}}}}{{{\eta _{iN}}}}} \right)\\
& + {g_i}{\rho _{i+1}} + {r_{i1}}z_{iN}^r + {r_{i2}}z_{iN}^m + {r_{i3}}z_{iN} +{{\hat d}_{iN}}\tanh \left( {\frac{{{z_{iN}}}}{{{\eta _{idN}}}}} \right)
\end{aligned}
\end{equation}
where
\begin{equation}\small
\begin{aligned}
{\dot {\hat \tau }_{iN}} = &{\delta _{i1N}}\left( {1 - {w_i}} \right){H_i}{z_{iN}}\tanh \left( {\frac{{{H_i}{z_{iN}}}}{{{\eta _{iN}}}}} \right) - {\delta _{i2N}}{{\hat \tau }_{iN}}^r\\
& - {\delta _{i3N}}\hat \tau _{iN}^m- {\delta _{i4N}}\hat \tau _{iN}\\
{\dot {\hat d}_{iN}} = &{q_{i1N}}{z_{iN}}\tanh \left( {\frac{{{z_{iN}}}}{{{\eta _{idN}}}}} \right) - {q_{i2N}}{{\hat d}_{iN}}^r - {q_{i3N}}\hat d_{iN}^m- {q_{i4N}}\hat d_{iN}
\end{aligned}
\end{equation}
\begin{theorem}
Considering the uncertain strict-feedback system (1), the compensation system (63), the prediction errors and modified F-FxT SPEM (74) or (78), the virtual control law and controller (72) or (76), the adaptive laws (73) or (77), all signals in the closed-loop system are globally F-FxT bounded, and the tracking error tends to an arbitrarily small domain near zero for fast fixed time. 
\end{theorem}
\begin{remark}
If we set ${H_i}\left( {{{\bar \rho }_i}} \right) = 1$, the control strategies we designed can be adopted to address the case where the boundaries of the uncertain terms are unknown. 
\end{remark}
%4-A
\section{Globally composite neural FnT/FxT controllers with input nonlinearities}
\subsection{Actuator magnitude and rate saturations with known parameters}
This section takes Method 4 as an example to illustrate the controller design process under input magnitude and rate (IMR) saturations, which is also applicable to Methods 1 to 9 and Methods 1-single to 9-single.

The IMR saturations are described as $\underline u \le  u \le \bar u,\underline \gamma  \le \dot u \le \bar \gamma$. Following [36], the IMR saturations are transformed into the following form
\begin{equation}
{\Psi _{\min }} \le u \le {\Psi _{\max }}
\end{equation}
where ${\Psi _{\min }} = \max \left\{ {u\left( {t - {T_r}} \right) + \underline \gamma {T_r},\underline u} \right\},{\Psi _{\max }} = \min \left\{ {u\left( {t - {T_r}} \right) + \bar \gamma {T_r},\bar u} \right\}$, and $T_r>0$ is the sample time.

Then, system (1) subject to IMR saturations is formulated as
\begin{equation}
\left\{  \begin{aligned}
&{{\dot \rho }_i} = {h_i}\left( {{{\bar \rho }_i}} \right) + {g_i}\left( {{{\bar \rho }_i}} \right){\rho _{i + 1}} + {d_i},i = 1,2, \ldots ,n - 1\\
&{{\dot \rho }_n} = {h_n}\left( {{{\bar \rho }_n}} \right) + {g_n}\left( {{{\bar \rho }_n}} \right)\Psi \left( u \right) + {d_n}\\
&y = {\rho _1}
 \end{aligned} \right.
\end{equation}
where
\begin{equation}
\Psi \left( u \right) = \left\{ 
\begin{aligned}
&{\Psi _{\min }}, u \le {\Psi _{\min }}\\
&u, {\Psi _{\min }} \le u \le {\Psi _{\max }}\\
&{\Psi _{\max }}, u \ge {\Psi _{\max }}
\end{aligned} 
\right.
\end{equation}

To work out the known IMR saturations, a novel adaptive F-FxT auxiliary variable system is designed to effectively diminish saturation time, and further alleviate the impact of the auxiliary signal on the tracking error, which is expressed as follows
\begin{equation}
\dot \Delta  =  - {c_1}{P_\Delta }{\mathop{\rm sig}\nolimits} {\left( \Delta  \right)^{{b_1}}} - {c_2}{P_\Delta }{\mathop{\rm sig}\nolimits} {\left( \Delta  \right)^{{b_2}}} - {c_3}{P_\Delta }\Delta  + {g_n}\left( {\Psi \left( u \right) - u} \right)
\end{equation}
where ${c_1} > 0,{c_2} > 0,{c_3} > 0,0 < {b_1} < 1,{b_2} > 1$, and ${P_\Delta }$ is the adaptive positive gain, which is expressed as
\begin{equation}\small
\begin{aligned}
&{\dot P_\Delta } = \left\{ \begin{aligned}
 &- {k_{p1}}P_\Delta ^{{{\bar b}_{p1}}} - {k_{p2}}P_\Delta ^{{{\bar b}_{p2}}} + {k_p}{P^L},&\Psi \left( u \right) \ne u\\
&{{\mathop{\rm Proj}\nolimits} _{{P_\Delta }}}\left( {{k_{p3}}P_\Delta ^{{b_3}} + {k_{p4}}P_\Delta ^{ - {b_4}} + {k_{p5}}{P_\Delta } + {c_p}} \right),&\Psi \left( u \right) = u
\end{aligned} \right. \\
&{{\mathop{\rm Proj}\nolimits} _{{P_\Delta }}}\left(  \cdot  \right) = \left\{ \begin{aligned}
&0,\left( {{P_\Delta } = {P^U}, \cdot  \ge 0} \right),or \left( {{P_\Delta } = {P^L}, \cdot  \le 0} \right)\\
&\cdot ,otherwise
\end{aligned} \right.
\end{aligned}
\end{equation}
where ${k_p},{k_{p1}},{k_{p2}},{k_{p3}},{k_{p4}},{k_{p5}},{b_3},{b_4},{c_p}$ are positive design parameters, and ${\bar b_{p1}} = b_{p1}^{{\mathop{\rm sgn}} \left( {{P_\Delta } - 1} \right)},{\bar b_{p2}} = b_{p2}^{{\mathop{\rm sgn}} \left( {1 - {P_\Delta }} \right)},{b_{p1}} > 1,0 < {b_{p2}} < 1$. ${P^L}$ and ${P^U}$ are the boundaries of ${P_\Delta }$
\begin{remark}
The $\Psi \left( u \right)$ can also be replaced by the smooth function in [36].
\end{remark}

The compensation tracking errors are defined as below
\begin{equation}
{\lambda _i} = {\zeta _i} - {\sigma _i},{\lambda _n} = {\zeta _n} - {\sigma _n} - \Delta ,i = 1,2, \ldots ,n - 1
\end{equation}

Then, the virtual control variables ${\alpha _i}$ and the control input $u$ are developed as follows (Method 10)
\begin{equation}\tiny
\begin{aligned}
{\alpha _1}& = \frac{1}{{{g_1}}}\left[ { - {k_{12}}\lambda _1^{r} + {{\dot y}_r} - {k_{11}}{\varphi _1}\left( {{\lambda _1}} \right) - {{\hat d}_1}\tanh \left( {\frac{{{\lambda _1}}}{{{\eta _{1d}}}}} \right)} \right]\\
 &- \frac{1}{{{g_1}}}\left[ {\frac{{{w_1}\left( {{{\bar \rho }_1}} \right){\lambda _1}}}{{2a_1^2}}{{\hat L}_1}\psi _1^T{\psi _1} + \left( {1 - {w_1}} \right){{\hat \tau }_1}{H_1}\tanh \left( {\frac{{{H_1}{\lambda _1}}}{{{\eta _1}}}} \right)} \right],\\
{\alpha _i} &= \frac{1}{{{g_i}}}\left[ { - {k_{i2}}\lambda _i^{r} + {{\dot \rho }_{i,c}} - {g_{i - 1}}{\zeta _{i - 1}} - {k_{i1}}{\varphi _i}\left( {{\lambda _i}} \right) - {{\hat d}_i}\tanh \left( {\frac{{{\lambda _i}}}{{{\eta _{id}}}}} \right)} \right]\\
 &- \frac{1}{{{g_i}}}\left[ {\frac{{{w_i}\left( {{{\bar \rho }_i}} \right){\lambda _i}}}{{2a_i^2}}{{\hat L}_i}\psi _i^T{\psi _i} + \left( {1 - {w_i}} \right){{\hat \tau }_i}{H_i}\tanh \left( {\frac{{{H_i}{\lambda _i}}}{{{\eta _i}}}} \right)} \right] ,(i = 2, \ldots ,n - 2),\\
{\alpha _i} &= \frac{1}{{{g_i}}}\left[ { - {k_{i2}}\lambda _i^{r} + {{\dot \rho }_{i,c}} - {g_{i - 1}}{\zeta _{i - 1}} - {k_{i1}}{\varphi _i}\left( {{\lambda _i}} \right) - {{\hat d}_i}\tanh \left( {\frac{{{\lambda _i}}}{{{\eta _{id}}}}} \right)} \right]\\
 &- \frac{1}{{{g_i}}}\left[ {\frac{{{w_i}\left( {{{\bar \rho }_i}} \right){\lambda _i}}}{{2a_i^2}}{{\hat L}_i}\psi _i^T{\psi _i} + \left( {1 - {w_i}} \right){{\hat \tau }_i}{H_i}\tanh \left( {\frac{{{H_i}{\lambda _i}}}{{{\eta _i}}}} \right)} \right]-\Delta ,(i=n-1),\\
u &= \frac{1}{{{g_n}}}\left[ { - {k_{n2}}\lambda _n^{r} + {{\dot \rho }_{n,c}} - {g_{n - 1}}{\zeta _{n - 1}} - {k_{n1}}{\varphi _n}\left( {{\lambda _n}} \right) - {{\hat d}_n}\tanh \left( {\frac{{{\lambda _n}}}{{{\eta _{nd}}}}} \right)} \right]\\
 &- \frac{1}{{{g_n}}}\left[ {\frac{{{w_n}\left( {{{\bar \rho }_n}} \right){\lambda _n}}}{{2a_n^2}}{{\hat L}_n}\psi _n^T{\psi _n} + \left( {1 - {w_n}} \right){{\hat \tau }_n}{H_n}\tanh \left( {\frac{{{H_n}{\lambda _n}}}{{{\eta _n}}}} \right)} \right]\\
&- \frac{1}{{{g_n}}}\left[ {{c_1}{P_\Delta }{\mathop{\rm sig}\nolimits} {{\left( \Delta  \right)}^{{b_1}}} + {c_2}{P_\Delta }{\mathop{\rm sig}\nolimits} {{\left( \Delta  \right)}^{{b_2}}} + {c_3}{P_\Delta }\Delta } \right],
\end{aligned}
\end{equation} 
 where ${\hat L_i},{\hat \tau _i},{\hat d_i}$ are as follows
\begin{equation}
\begin{aligned}
&{\dot{ \hat L}_i} = {\beta _{hi}}{w_i}\left( {\lambda _i^2 + {\beta _{zi}}z_{iN}^2} \right)\psi _i^T{\psi _i} - {\beta _{i1}}{{\hat L}_i}^r - {\beta _{i2}}\hat L_i^m\\
&{\dot {\hat \tau }_i} = {\delta _{i1}}\left( {1 - {w_i}} \right){H_i}{\lambda _i}\tanh \left( {\frac{{{H_i}{\lambda _i}}}{{{\eta _i}}}} \right) - {\delta _{i2}}{{\hat \tau }_i}^r - {\delta _{i3}}\hat \tau _i^m\\
&{\dot {\hat d}_i} = {q_{i1}}{\lambda _i}\tanh \left( {\frac{{{\lambda _i}}}{{{\eta _{id}}}}} \right) - {q_{i2}}{{\hat d}_i}^r - {q_{i3}}\hat d_i^m
\end{aligned}
\end{equation}

The prediction errors ${z_{iN}} = {\rho _i} - {\hat \rho _i}$ are extracted from the following modified FxT SPEM (denote ${\rho_{n+1}}=\Psi \left( u \right)$):
\begin{equation}
\begin{aligned}
{\dot{\hat \rho }_i} =& \frac{{{w_i}\left( {{{\bar \rho }_i}} \right)}}{{2a_i^2}}{{\hat L}_i}{z_{iN}}\psi _i^T{\psi _i} + \left( {1 - {w_i}} \right){{\hat \tau }_{iN}}{H_i}\tanh \left( {\frac{{{H_i}{z_{iN}}}}{{{\eta _{iN}}}}} \right)\\
 &+ {g_i}{\rho _{i+1}} + {r_{i1}}z_{iN}^r + {r_{i2}}z_{iN}^m + {{\hat d}_{iN}}\tanh \left( {\frac{{{z_{iN}}}}{{{\eta _{idN}}}}} \right)
\end{aligned}
\end{equation}
where
\begin{equation}\small
\begin{aligned}
&{\dot {\hat \tau }_{iN}} = {\delta _{i1N}}\left( {1 - {w_i}} \right){H_i}{z_{iN}}\tanh \left( {\frac{{{H_i}{z_{iN}}}}{{{\eta _{iN}}}}} \right) - {\delta _{i2N}}{{\hat \tau }_{iN}}^r - {\delta _{i3N}}\hat \tau _{iN}^m\\
&{\dot {\hat d}_{iN}} = {q_{i1N}}{z_{iN}}\tanh \left( {\frac{{{z_{iN}}}}{{{\eta _{idN}}}}} \right) - {q_{i2N}}{{\hat d}_{iN}}^r - {q_{i3N}}\hat d_{iN}^m
\end{aligned}
\end{equation}
\begin{theorem}
Considering the system (81) under bounded $u$, the compensation system (39), the adaptive F-FxT auxiliary variable system (83), the prediction errors (88), and the virtual control law and controller (86), all signals in the closed-loop system are globally FxT bounded, and the tracking error tends to an arbitrarily small domain near zero for fixed time. 
\end{theorem}
\begin{remark}
The developed Method 10 provides a framework for implementing global composite neural FnT/FxT control under IMR saturations, and the auxiliary signal compensators in \cite{sat_a1,sat_a2,sat_f1,sat_f2,sat_fa} can also be applied to this.
\end{remark}
%4-B
\subsection{Multiple input nonlinearities with unknown parameters}
This section takes Method 4 and Method 4-single as examples to illuminate the controller design process under unknown input magnitude saturation, rate saturation and dead-zone simultaneously, which is also applicable to Methods 1 to 9 and Methods 1-single to 9-single.

By adopting (80), the input subject to unknown magnitude saturation, rate saturation and dead-zone is formulated as
\begin{equation}
\Psi \left( u \right) = \left\{ 
\begin{aligned}
&{\Psi _{\min }},u \le {c_{s1}}\\
&{t_s}\left( u \right),{c_{s1}} \le u \le {c_{s2}}\\
&0,{c_{s2}} \le u \le {c_{r1}}\\
&{t_r}\left( u \right),{c_{r1}} \le u \le {c_{r2}}\\
&{\Psi _{\max }},u \ge {c_{r2}}
\end{aligned} \right.
\end{equation} 
where the assumption of ${t_s}\left( u \right),{t_r}\left( u \right)$  is the same as [41], and ${c_{s1}} < {c_{s2}} < 0,0 < {c_{r1}} < {c_{r2}}$ are unknown parameters.

Enlightened by [41], by introducing a smooth function to model the non-smooth actuator nonlinearities, $\Psi \left( u \right)$ can be rewritten as 
\begin{equation}
\Psi \left( u \right) = {g_u} \cdot u + P\left( u \right)
\end{equation} 
with $0 < b < {g_u} < \bar b,\left| {P\left( u \right)} \right| \le \bar P$.

Define tracking errors as follows
\begin{equation}
\begin{aligned}
&{\zeta _1} = {\rho _1} - {y_r}\\
&{\zeta _i} = {\rho _i} - {\rho _{i,c}},\quad i = 2,3, \ldots ,n
\end{aligned}
\end{equation}
where ${\rho _{i,c}}$ are obtained by the FxT command filter \cite{2011Uniform} with virtual control variables ${\alpha _{i - 1}}$ being inputs.

To alleviate the impact of $\left( {{\rho _{i,c}} - {\alpha _{i - 1}}} \right)$, the following compensation system is designed
\begin{equation}
\begin{aligned}
&{{\dot \sigma }_1} =  - {k_{11}}\sigma _1^m - {k_{12}}\sigma _1^r + {g_1}\left( {{\rho _{2,c}} - {\alpha _1}} \right) + {g_1}{\sigma _2}\\
&{{\dot \sigma }_i} =  - {k_{i1}}\sigma _i^m - {k_{i2}}\sigma _i^r + {g_i}\left( {{\rho _{\left( {i + 1} \right),c}} - {\alpha _i}} \right) - {g_{i - 1}}{\sigma _{i - 1}} + {g_i}{\sigma _{i + 1}}\\
&{{\dot \sigma }_n} =  0
\end{aligned}
\end{equation}

Then, the virtual control variables ${\alpha _i}$ and the control input $u$ are established as (Method 11)
\begin{equation}\small
\begin{aligned}
&{\alpha _1} = \frac{1}{{{g_1}}}\left[ { - {k_{12}}\lambda _1^{r} + {{\dot y}_r} - {k_{11}}{\varphi _1}\left( {{\lambda _1}} \right) - {{\hat d}_1}\tanh \left( {\frac{{{\lambda _1}}}{{{\eta _{1d}}}}} \right)} \right]\\
 &- \frac{1}{{{g_1}}}\left[ {\frac{{{w_1}\left( {{{\bar \rho }_1}} \right){\lambda _1}}}{{2a_1^2}}{{\hat L}_1}\psi _1^T{\psi _1} + \left( {1 - {w_1}} \right){{\hat \tau }_1}{H_1}\tanh \left( {\frac{{{H_1}{\lambda _1}}}{{{\eta _1}}}} \right)} \right],\\
&{\alpha _i} = \frac{1}{{{g_i}}}\left[ { - {k_{i2}}\lambda _i^{r} + {{\dot \rho }_{i,c}} - {g_{i - 1}}{\zeta _{i - 1}} - {k_{i1}}{\varphi _i}\left( {{\lambda _i}} \right) - {{\hat d}_i}\tanh \left( {\frac{{{\lambda _i}}}{{{\eta _{id}}}}} \right)} \right]\\
 &- \frac{1}{{{g_i}}}\left[ {\frac{{{w_i}\left( {{{\bar \rho }_i}} \right){\lambda _i}}}{{2a_i^2}}{{\hat L}_i}\psi _i^T{\psi _i} + \left( {1 - {w_i}} \right){{\hat \tau }_i}{H_i}\tanh \left( {\frac{{{H_i}{\lambda _i}}}{{{\eta _i}}}} \right)} \right],(i = 2, \ldots ,n - 1),\\
&u = \frac{1}{{{g_n}}}\left[ { - {k_{n2}}\lambda _n^{r} - {k_{n1}}{\varphi _n}\left( {{\lambda _n}} \right) - {{\hat d}_n}\tanh \left( {\frac{{{\lambda _n}}}{{{\eta _{nd}}}}} \right)} \right]\\
 &- \frac{1}{{{g_n}}}\left[ {\frac{{{w_n}\left( {{{\bar \rho }_n}} \right){\lambda _n}}}{{2a_n^2}}{{\hat L}_n}\psi _n^T{\psi _n} + \left( {1 - {w_n}} \right){{\hat \tau }_n}\tanh \left( {\frac{{{\lambda _n}}}{{{\eta _n}}}} \right)} \right],
\end{aligned}
\end{equation} 
where ${\hat L_i},{\hat \tau _i},{\hat d_i}$ are as follows
\begin{equation}
\begin{aligned}
{\dot{ \hat L}_i}=  &{\beta _{hi}}{w_i}\left( {\lambda _i^2 + {\beta _{zi}}z_{iN}^2} \right)\psi _i^T{\psi _i} - {\beta _{i1}}{{\hat L}_i}^r \\
&- {\beta _{i2}}\hat L_i^m,(i = 1, \ldots ,n - 1),\\
{\dot{ \hat L}_n} = &{\beta _{hn}}{w_u}{\lambda _n^2}\psi _n^T{\psi _n} - {\beta _{n1}}{{\hat L}_n}^r - {\beta _{n2}}\hat L_n^m,\\
{\dot {\hat \tau }_i} = &{\delta _{i1}}\left( {1 - {w_i}} \right){H_i}{\lambda _i}\tanh \left( {\frac{{{H_i}{\lambda _i}}}{{{\eta _i}}}} \right) - {\delta _{i2}}{{\hat \tau }_i}^r \\
&- {\delta _{i3}}\hat \tau _i^m,(i = 1, \ldots ,n - 1),\\
{\dot {\hat \tau }_n} = &{\delta _{n1}}\left( {1 - {w_n}} \right){\lambda _n}\tanh \left( {\frac{{{\lambda _n}}}{{{\eta _n}}}} \right) - {\delta _{n2}}{{\hat \tau }_n}^r - {\delta _{n3}}\hat \tau _n^m\\
{\dot {\hat d}_i} = &{q_{i1}}{\lambda _i}\tanh \left( {\frac{{{\lambda _i}}}{{{\eta _{id}}}}} \right) - {q_{i2}}{{\hat d}_i}^r - {q_{i3}}\hat d_i^m
\end{aligned}
\end{equation}

The prediction errors ${z_{iN}} = {\rho _i} - {\hat \rho _i},(i = 1, \ldots ,n - 1)$ are extracted from the following modified FxT SPEM:
\begin{equation}
\begin{aligned}
{\dot{\hat \rho }_i} =& \frac{{{w_i}\left( {{{\bar \rho }_i}} \right)}}{{2a_i^2}}{{\hat L}_i}{z_{iN}}\psi _i^T{\psi _i} + \left( {1 - {w_i}} \right){{\hat \tau }_{iN}}{H_i}\tanh \left( {\frac{{{H_i}{z_{iN}}}}{{{\eta _{iN}}}}} \right)\\
 &+ {g_i}{\rho _{i+1}} + {r_{i1}}z_{iN}^r + {r_{i2}}z_{iN}^m + {{\hat d}_{iN}}\tanh \left( {\frac{{{z_{iN}}}}{{{\eta _{idN}}}}} \right)
\end{aligned}
\end{equation}
where
\begin{equation}\small
\begin{aligned}
&{\dot {\hat \tau }_{iN}} = {\delta _{i1N}}\left( {1 - {w_i}} \right){H_i}{z_{iN}}\tanh \left( {\frac{{{H_i}{z_{iN}}}}{{{\eta _{iN}}}}} \right) - {\delta _{i2N}}{{\hat \tau }_{iN}}^r - {\delta _{i3N}}\hat \tau _{iN}^m\\
&{\dot {\hat d}_{iN}} = {q_{i1N}}{z_{iN}}\tanh \left( {\frac{{{z_{iN}}}}{{{\eta _{idN}}}}} \right) - {q_{i2N}}{{\hat d}_{iN}}^r - {q_{i3N}}\hat d_{iN}^m
\end{aligned}
\end{equation}

To further decline the updating weights, we define $L = \max \left\{ {{{\left\| {{l_i}} \right\|}^2}} \right\},\left( {i = 1,2, \ldots ,n-1} \right)$. 

Then, the virtual control variables ${\alpha _i}$ and the control input $u$ are developed as follows (Method 12)
\begin{equation}\small
\begin{aligned}
&{\alpha _1} = \frac{1}{{{g_1}}}\left[ { - {k_{12}}\lambda _1^{r} + {{\dot y}_r} - {k_{11}}{\varphi _1}\left( {{\lambda _1}} \right) - {{\hat d}_1}\tanh \left( {\frac{{{\lambda _1}}}{{{\eta _{1d}}}}} \right)} \right]\\
 &- \frac{1}{{{g_1}}}\left[ {\frac{{{w_1}{\lambda _1}}}{{2a_1^2}}{{\hat L}}\psi _1^T{\psi _1} + \left( {1 - {w_1}} \right){{\hat \tau }_1}{H_1}\tanh \left( {\frac{{{H_1}{\lambda _1}}}{{{\eta _1}}}} \right)} \right]\\
&{\alpha _i} = \frac{1}{{{g_i}}}\left[ { - {k_{i2}}\lambda _i^{r} + {{\dot \rho }_{i,c}} - {g_{i - 1}}{\zeta _{i - 1}} - {k_{i1}}{\varphi _i}\left( {{\lambda _i}} \right) - {{\hat d}_i}\tanh \left( {\frac{{{\lambda _i}}}{{{\eta _{id}}}}} \right)} \right]\\
 &- \frac{1}{{{g_i}}}\left[ {\frac{{{w_i}{\lambda _i}}}{{2a_i^2}}{{\hat L}}\psi _i^T{\psi _i} + \left( {1 - {w_i}} \right){{\hat \tau }_i}{H_i}\tanh \left( {\frac{{{H_i}{\lambda _i}}}{{{\eta _i}}}} \right)} \right],(i = 2, \ldots ,n - 1),\\
&u = \frac{1}{{{g_n}}}\left[ { - {k_{n2}}\lambda _n^{r} - {k_{n1}}{\varphi _n}\left( {{\lambda _n}} \right) - {{\hat d}_n}\tanh \left( {\frac{{{\lambda _n}}}{{{\eta _{nd}}}}} \right)} \right]\\
 &- \frac{1}{{{g_n}}}\left[ {\frac{{{w_n}\left( {{{\bar \rho }_n}} \right){\lambda _n}}}{{2a_n^2}}{{\hat L}_u}\psi _n^T{\psi _n} + \left( {1 - {w_n}} \right){{\hat \tau }_n}\tanh \left( {\frac{{{\lambda _n}}}{{{\eta _n}}}} \right)} \right],
\end{aligned}
\end{equation} 
where ${\hat L},{\hat \tau _i},{\hat d_i}$ are as follows
\begin{equation}
\begin{aligned}
\dot {\hat L} = &{\beta _{h}}\sum\limits_{i = 1}^{n-1} {\frac{{{w_i}}}{{2a_i^2}}\left( {\lambda _i^2 + {\beta _{zi}}z_{iN}^2} \right)\psi _i^T{\psi _i}}  - {\beta _{1}}{\hat L^r} - {\beta _{2}}{\hat L^m}\\
{\dot{ \hat L}_u} = &{\beta _{hn}}{w_u}{\lambda _n^2}\psi _n^T{\psi _n} - {\beta _{n1}}{{\hat L}_u}^r - {\beta _{n2}}\hat L_u^m\\
{\dot {\hat \tau }_i} = &{\delta _{i1}}\left( {1 - {w_i}} \right){H_i}{\lambda _i}\tanh \left( {\frac{{{H_i}{\lambda _i}}}{{{\eta _i}}}} \right) - {\delta _{i2}}{{\hat \tau }_i}^r \\
&- {\delta _{i3}}\hat \tau _i^m,(i = 1, \ldots ,n - 1),\\
{\dot {\hat \tau }_n} = &{\delta _{n1}}\left( {1 - {w_n}} \right){\lambda _n}\tanh \left( {\frac{{{\lambda _n}}}{{{\eta _n}}}} \right) - {\delta _{n2}}{{\hat \tau }_n}^r - {\delta _{n3}}\hat \tau _n^m\\
{\dot {\hat d}_i} = &{q_{i1}}{\lambda _i}\tanh \left( {\frac{{{\lambda _i}}}{{{\eta _{id}}}}} \right) - {q_{i2}}{{\hat d}_i}^r - {q_{i3}}\hat d_i^m
\end{aligned}
\end{equation}

The prediction errors ${z_{iN}} = {\rho _i} - {\hat \rho _i},(i = 1, \ldots ,n - 1)$ are extracted from the following improved fixed-time SPEM:
\begin{equation}
\begin{aligned}
{\dot{\hat \rho }_i} =& \frac{{{w_i}\left( {{{\bar \rho }_i}} \right)}}{{2a_i^2}}{{\hat L}}{z_{iN}}\psi _i^T{\psi _i} + \left( {1 - {w_i}} \right){{\hat \tau }_{iN}}{H_i}\tanh \left( {\frac{{{H_i}{z_{iN}}}}{{{\eta _{iN}}}}} \right)\\
 &+ {g_i}{\rho _{i+1}} + {r_{i1}}z_{iN}^r + {r_{i2}}z_{iN}^m + {{\hat d}_{iN}}\tanh \left( {\frac{{{z_{iN}}}}{{{\eta _{idN}}}}} \right)
\end{aligned}
\end{equation}
where
\begin{equation}\small
\begin{aligned}
&{\dot {\hat \tau }_{iN}} = {\delta _{i1N}}\left( {1 - {w_i}} \right){H_i}{z_{iN}}\tanh \left( {\frac{{{H_i}{z_{iN}}}}{{{\eta _{iN}}}}} \right) - {\delta _{i2N}}{{\hat \tau }_{iN}}^r - {\delta _{i3N}}\hat \tau _{iN}^m\\
&{\dot {\hat d}_{iN}} = {q_{i1N}}{z_{iN}}\tanh \left( {\frac{{{z_{iN}}}}{{{\eta _{idN}}}}} \right) - {q_{i2N}}{{\hat d}_{iN}}^r - {q_{i3N}}\hat d_{iN}^m
\end{aligned}
\end{equation}
\begin{theorem}
Considering the system (81) subject to (90), the compensation system (93), the prediction errors (96) or (100), the virtual control law and controller (94) or (98), all signals in the closed-loop system are globally FxT bounded, and the tracking error tends to an arbitrarily small domain near zero for fixed time. 
\end{theorem}
\begin{remark}
The established Method 11 and Method 12 can also be applied to achieve global composite neural FnT/FxT control under both input saturation and dead-zone, IMR saturations, or single input nonlinearity.
\end{remark}
\begin{figure}
	\centering
	\subfigure[Tracking errors]{
	\includegraphics[width=0.22\textwidth]{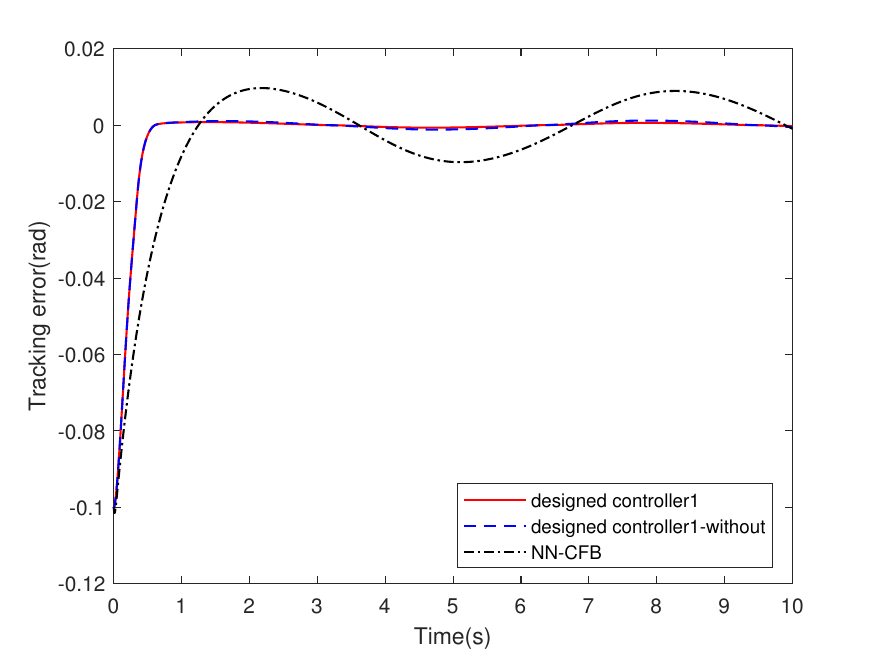}}
	\subfigure[Switching signal]{
	\includegraphics[width=0.22\textwidth]{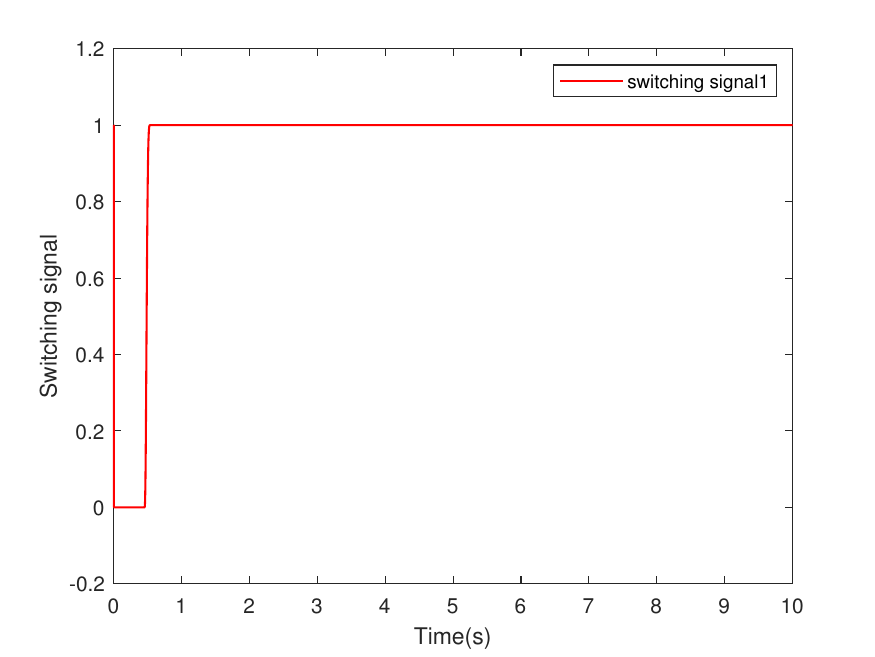}}
	\caption{Results of Method 1}
\end{figure}
\begin{figure}
	\centering
	\subfigure[Tracking errors]{
	\includegraphics[width=0.22\textwidth]{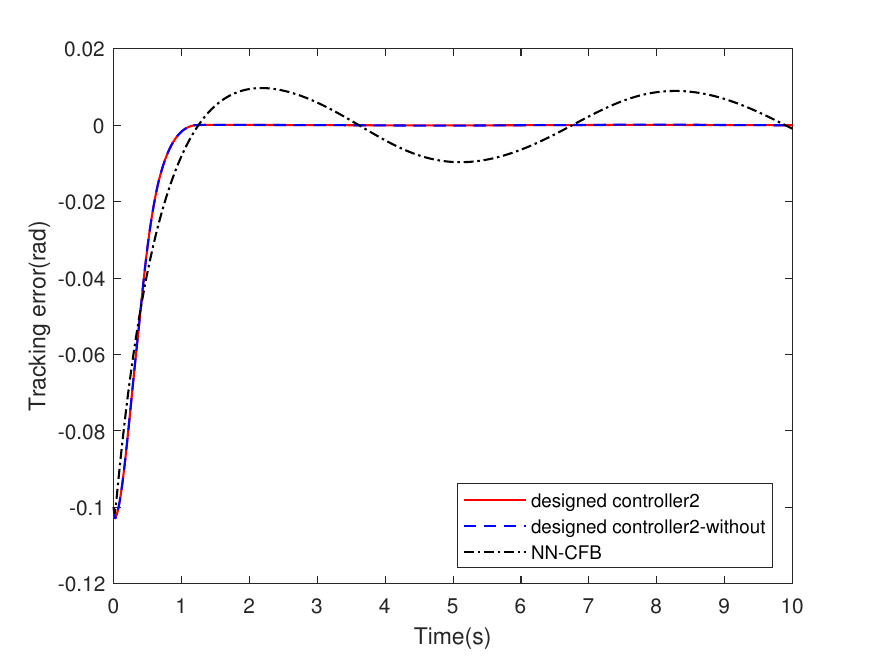}}
	\subfigure[Switching signal]{
	\includegraphics[width=0.22\textwidth]{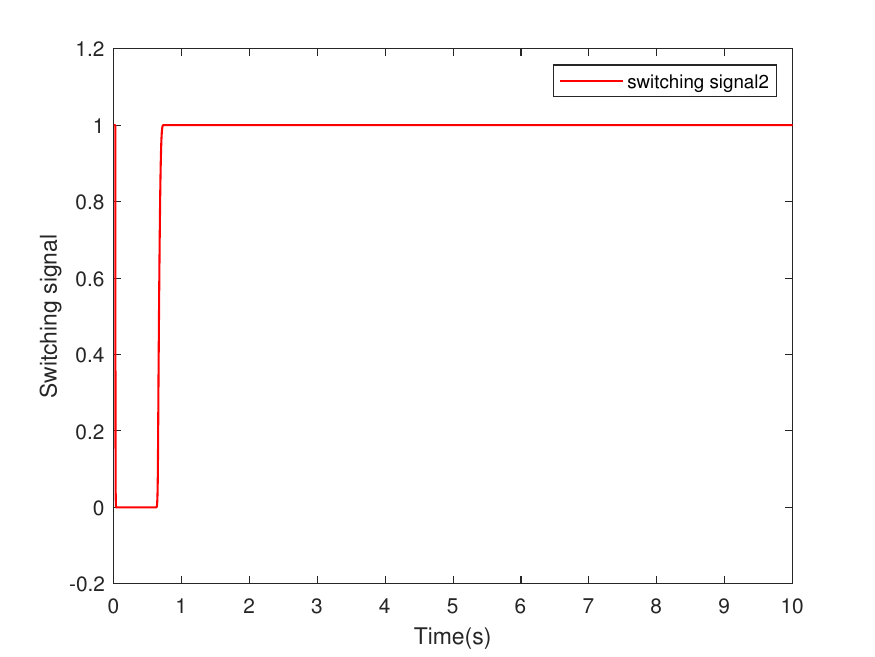}}
	\caption{Results of Method 2}
\end{figure}
\begin{figure}
	\centering
	\subfigure[Tracking errors]{
	\includegraphics[width=0.22\textwidth]{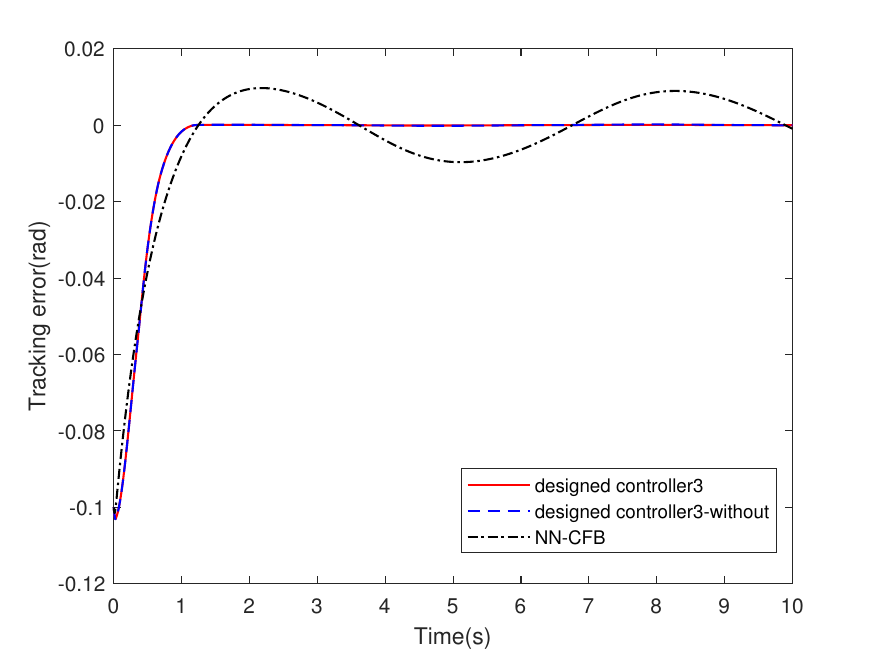}}
	\subfigure[Switching signal]{
	\includegraphics[width=0.22\textwidth]{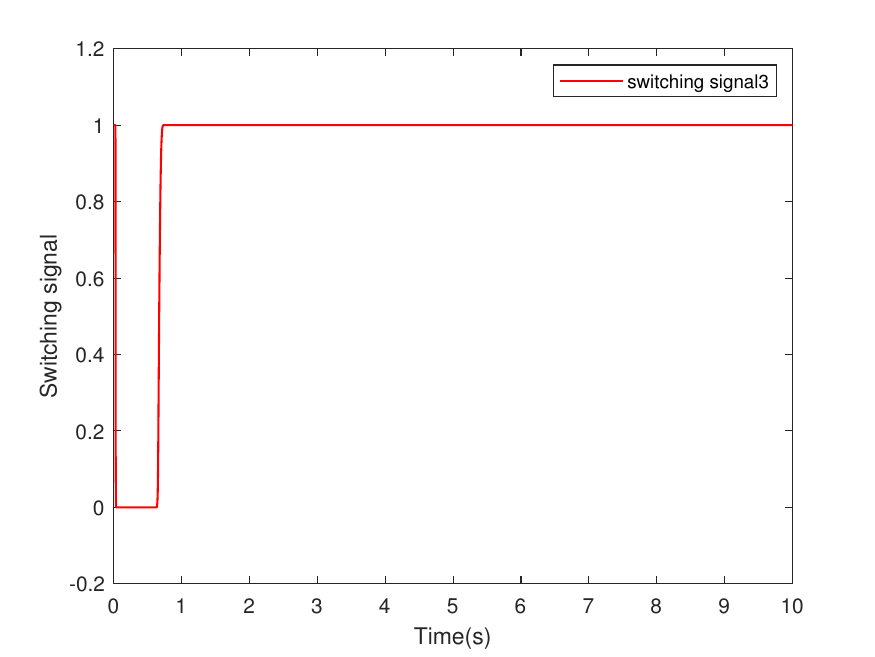}}
	\caption{Results of Method 3}
\end{figure}
\begin{figure}
	\centering
	\subfigure[Tracking errors]{
	\includegraphics[width=0.22\textwidth]{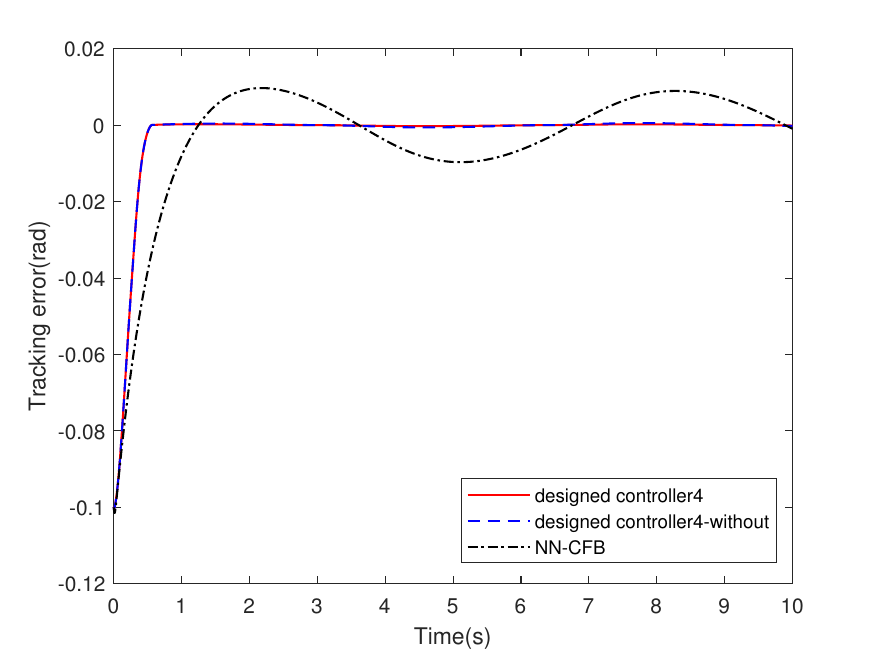}}
	\subfigure[Switching signal]{
	\includegraphics[width=0.22\textwidth]{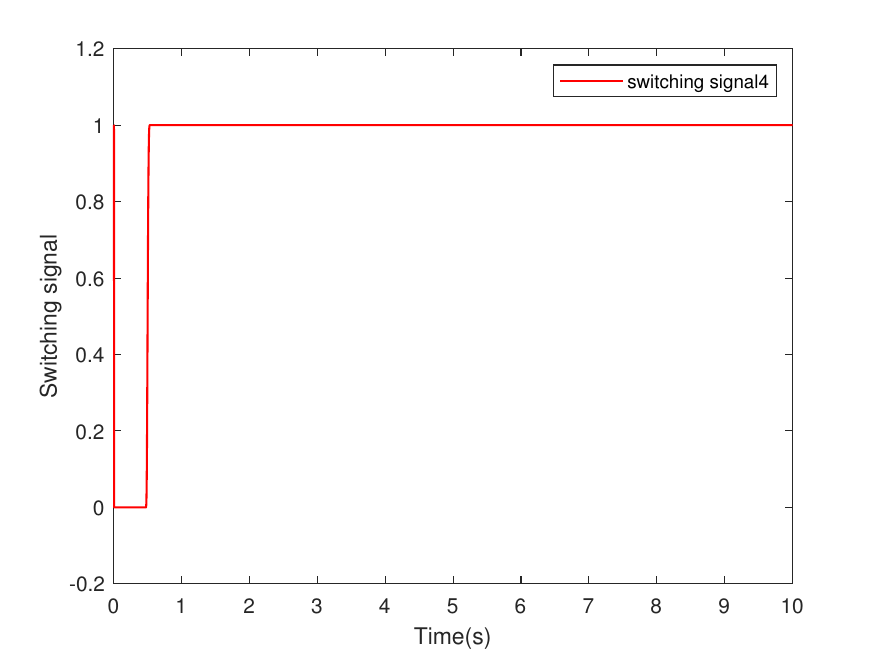}}
	\caption{Results of Method 4}
\end{figure}
\begin{figure}
	\centering
	\subfigure[Tracking errors]{
	\includegraphics[width=0.22\textwidth]{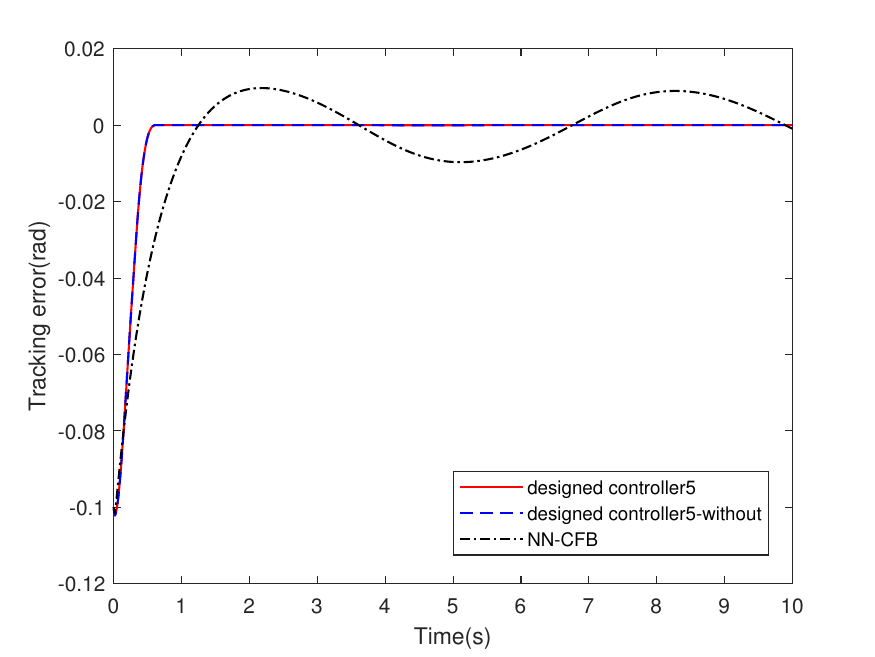}}
	\subfigure[Switching signal]{
	\includegraphics[width=0.22\textwidth]{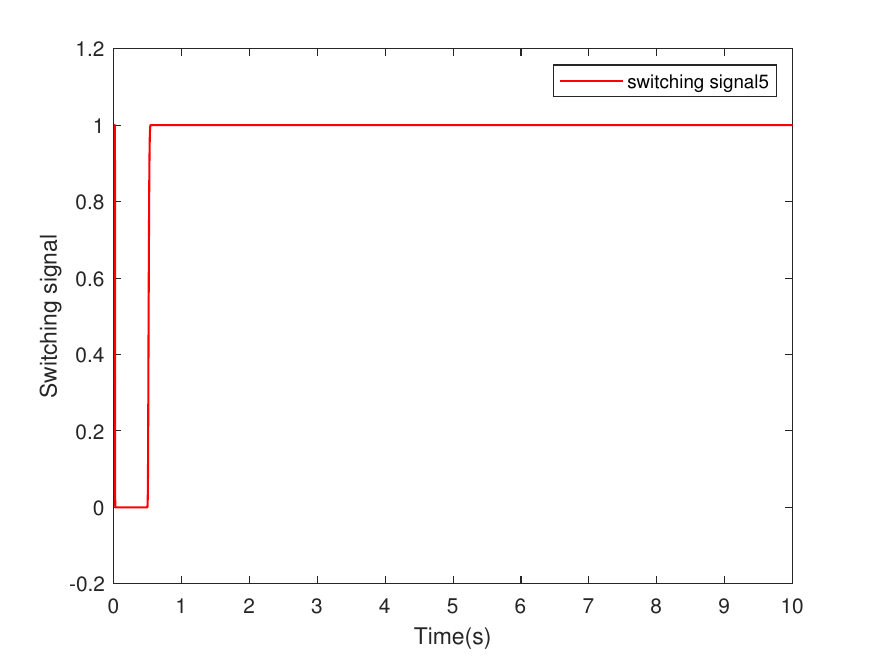}}
	\caption{Results of Method 5}
\end{figure}
\begin{figure}
	\centering
	\subfigure[Tracking errors]{
	\includegraphics[width=0.2\textwidth]{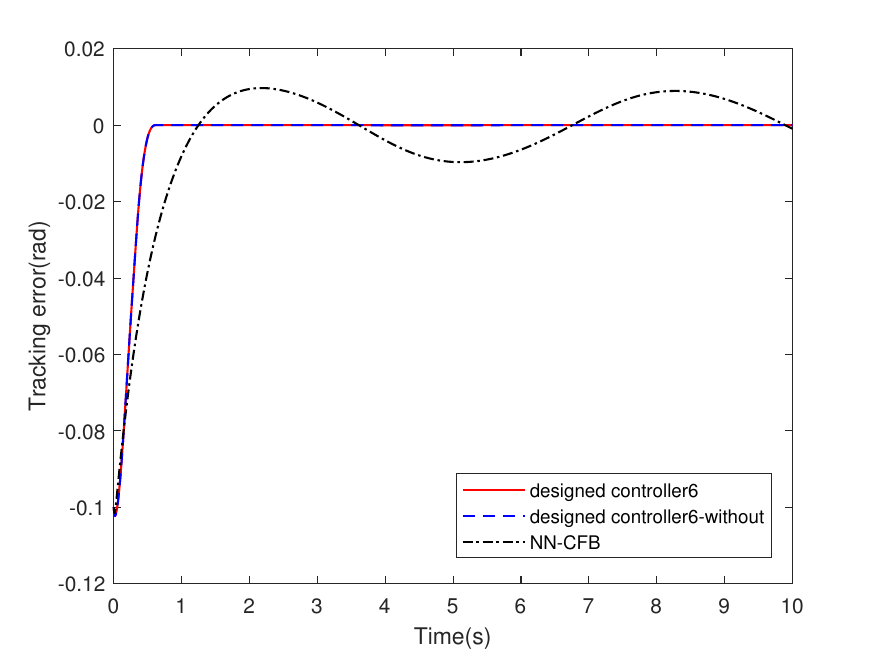}}
	\subfigure[Switching signal]{
	\includegraphics[width=0.22\textwidth]{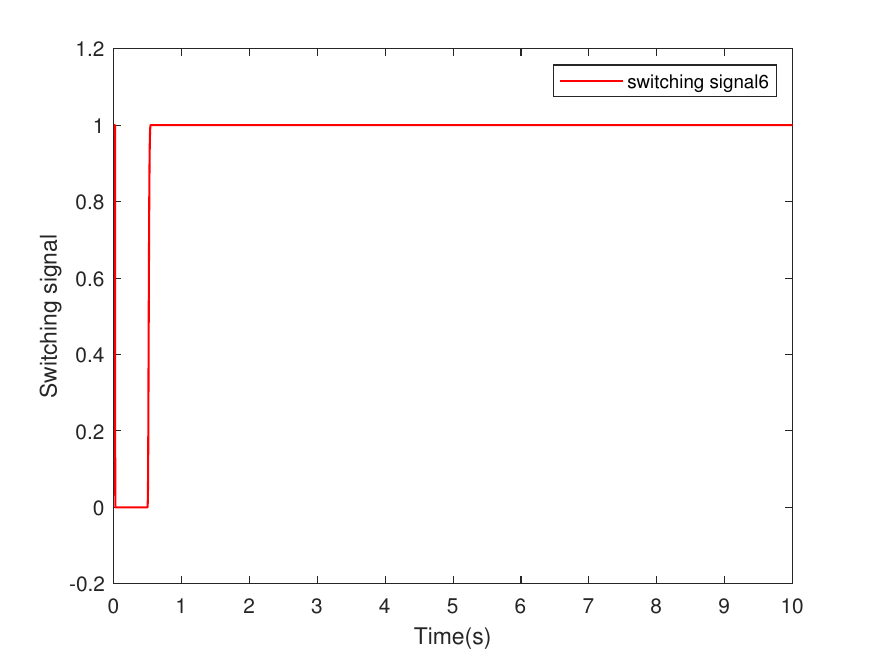}}
	\caption{Results of Method 6}
\end{figure}
%5-section
\section{Simulation Results}
In this section, the following inverted pendulum model \cite{inverted} is employed to evince validity of the designed control strategies:
\begin{equation}
\begin{aligned}
&{{\dot \rho }_1} = {\rho _2}\\
&{{\dot \rho }_2} = \frac{{{g_e}\sin {\rho _1} - \frac{{{m_a}{l_a}\rho _2^2\cos {\rho _1}\sin {\rho _1}}}{{{m_c} + {m_a}}}}}{{{l_a}\left( {\frac{4}{3} - \frac{{{m_a}{{\cos }^2}{\rho _1}}}{{{m_c} + {m_a}}}} \right)}} + \frac{{\frac{{\cos {\rho _1}}}{{{m_c} + {m_a}}}}}{{{l_a}\left( {\frac{4}{3} - \frac{{{m_a}{{\cos }^2}{\rho _1}}}{{{m_c} + {m_a}}}} \right)}}u + {d_2}
\end{aligned}
\end{equation}
where $g_e = 9.81\, {\rm{ }}{{\rm{m}} \mathord{\left/{\vphantom {{\rm{m}} {{{\rm{s}}^2}}}} \right.\kern-\nulldelimiterspace} {{{\rm{s}}^2}}}$, ${m_c} = 1\,{\rm{ kg}}$, ${m_a} = 0.1\,{\rm{ kg}}$, ${l_a} = 0.5\,{\rm{m}}$, and ${\rho _1}\left( 0 \right) =-0.1\,{\rm{ rad}}$. $y_r=0.2 \sin (t)\,{\rm{ rad}}$.

The parameters of designed FnT controllers are set as ${k_1} = 2,{k_2} = 4,m = 0.6,{p_1} = {p_2} = 0.5,{\gamma _1} = {\gamma _2} = 1,{a_2} = 1,{\beta _{h2}} = 10,{\delta _{21}} = {\delta _{21N}} = 10$, and the parameters of designed FxT controllers are selected as ${k_{11}}={k_{21}} = 2,{k_{12}}={k_{22}} = 4,r = 5/3,{a_2} = 1,{\beta _{h2}} = 10,{\delta _{21}} = {\delta _{21N}} = 10$. The RBFNN embodies 121 nodes with centers distributed evenly in $\left[ { - 0.25,0.25} \right] \times \left[ { - 0.25,0.25} \right]$ and widths are chosen as $2$. The designed controller without composite adaptive laws and the NN-based command-filtered backstepping algorithm [5] are adopted as the comparison approaches, labeled as designed controller-without and NN-CFB, respectively.

The simulation results displayed in Fig. 1-6 (a) manifest that the tracking performance is enhanced via the developed composite learning approaches, and the presented composite FnT/FxT controllers possess faster transient responses. Moreover, the switching signals depicted in Fig. 1-6 (b) reveal that the constructed robust control law has the capability to drive external transients back into the RBFNN's valid domain to ensure the global stability of the designed control strategies. 
\section{Conclusion}
In this article, globally composite adaptive law-based intelligent FnT/FxT control strategies have been developed to cope with the tracking control issue for a class of uncertain strict-feedback nonlinear systems. RBFNNs with novel composite updating weights are constructed to boost tracking performance and decline the NNs' learning parameters simultaneously. By virtue of a newly designed switching mechanism, system states can be brought back when they transcend the neural working region. A modified FnT/FxT backstepping algorithm is established to steer output to track the reference signal. A novel adaptive F-FxT auxiliary variable system is designed to cope with the input issue under known magnitude and rate saturations. For the input subject to unknown magnitude saturation, rate saturation and dead-zone simultaneously, a smooth function is introduced to model the non-smooth actuator nonlinearities, and the unknown parameters are approximated by adopting the proposed global RBFNNs. Stability analysis illuminates that the tracking error tends to an arbitrarily small residual set for finite/fixed time, while all signals in the closed-loop system are globally FnT/FxT bounded. A simulation example elucidates the validity of theoretical results.

% References
\balance
\bibliographystyle{Bibliography/IEEEtranTIE}
\bibliography{Bibliography/IEEEabrv,Bibliography/myRef}\ %IEEEabrv instead of IEEEfull

\end{document}